\documentclass[superscriptaddress,amsmath,amssymb,aps,pra,letterpaper,tightenlines,notitlepages,12pt]{revtex4-1}

\pdfoutput=1

\usepackage{amsmath}
\usepackage{amsfonts}
\usepackage{amssymb}
\usepackage{mathtools}
\usepackage{graphicx}
\usepackage{hyperref}
\usepackage{xfrac}
\usepackage{xcolor}

\providecommand{\U}[1]{\protect\rule{.1in}{.1in}}

\hypersetup{colorlinks=true,citecolor=blue,linkcolor=blue,filecolor=blue,urlcolor=blue,breaklinks=true}
\newtheorem{theorem}{Theorem}

\newtheorem{corollary}[theorem]{Corollary}

\newtheorem{definition}[theorem]{Definition}

\newtheorem{lemma}[theorem]{Lemma}

\newtheorem{remark}[theorem]{Remark}

\newenvironment{proof}[1][Proof]{\noindent\textbf{#1.} }{\ \rule{0.5em}{0.5em}}

\newcommand{\red}[1]{{\color{black}{#1}}}

\newcommand{\ket}[1]{\lvert #1 \rangle}
\newcommand{\bra}[1]{\langle #1 \rvert}

\allowdisplaybreaks

\begin{document}

\title{
Sample-based Hamiltonian and Lindbladian simulation: Non-asymptotic analysis of sample complexity}

\author{Byeongseon Go}
\affiliation{Department of Physics and Astronomy, Seoul National University, Seoul 08826, Republic of Korea}
\author{Hyukjoon Kwon}
\affiliation{School of Computational Sciences, Korea Institute for Advanced Study, Seoul 02455, Republic of Korea}
\author{Siheon Park}
\affiliation{Department of Physics and Astronomy, Seoul National University, Seoul 08826, Republic of Korea}
\author{Dhrumil Patel}
\affiliation{Department of Computer Science, Cornell University, Ithaca, New York 14850, USA}
\author{Mark M. Wilde}
\affiliation{School of Electrical and Computer Engineering, Cornell University, Ithaca, New York 14850, USA}

\begin{abstract}
Density matrix exponentiation (DME) is a quantum algorithm that processes multiple copies of a program state $\sigma$ to realize the Hamiltonian evolution $e^{-i \sigma t}$.
Wave matrix Lindbladization (WML) similarly processes multiple copies of a program state $\psi_L$ in order to realize a Lindbladian evolution. Both algorithms are prototypical sample-based quantum algorithms and can be used for various quantum information processing tasks, including quantum principal component analysis, Hamiltonian simulation, and Lindbladian simulation. In this work, we present detailed sample complexity analyses for DME and sample-based Hamiltonian simulation, as well as for WML and sample-based Lindbladian simulation. In particular, we prove that the sample complexity of DME is no larger than $4t^2/\varepsilon$, where~$t$ is the desired evolution time and~$\varepsilon$ is the desired imprecision level, as quantified by the normalized diamond distance. We also establish a fundamental lower bound on the sample complexity of sample-based Hamiltonian simulation, which matches our DME sample complexity bound up to a constant multiplicative factor, thereby proving that DME is optimal for sample-based Hamiltonian simulation. Additionally, we prove that the sample complexity of WML is no larger than $3t^2d^2/\varepsilon$, where $d$ is the dimension of the space on which the Lindblad operator acts nontrivially, and we prove a lower bound of $10^{-4} t^2/\varepsilon$ on the sample complexity of sample-based Lindbladian simulation. These results prove that WML is optimal for sample-based Lindbladian simulation whenever the Lindblad operator acts nontrivially on a constant-sized system.  Finally, we point out that the DME sample complexity analysis in [Kimmel \textit{et al}., npj Quantum Information 3, 13 (2017)] and the WML sample complexity analysis in [Patel and Wilde, Open Systems \& Information Dynamics 30, 2350010 (2023)] appear to be incomplete, highlighting the need for the results presented here.
\end{abstract}

\maketitle

\tableofcontents

\section{Introduction}

\subsection{Density matrix exponentiation}

Density matrix exponentiation (DME)
is a quantum algorithm that exponentiates a quantum state when multiple copies of it are available~\cite{lloyd2014quantum}. 
More specifically, given an unknown input quantum state $\rho$ and a sufficient number of copies of a ``program'' quantum state~$\sigma$, DME approximately implements the unitary evolution $e^{-i\sigma t}$ on $\rho$ for a given evolution time $t$. DME is thus a particular method for achieving the task of sample-based Hamiltonian simulation~\cite{Kimmel2017}, in which one realizes the unitary transformation $e^{-i\sigma t}$ approximately by consuming copies of a program state $\sigma$.

This exponentiation of quantum states offers a powerful tool for various quantum information processing tasks.
As initially proposed in Ref.~\cite{lloyd2014quantum}, it can be employed for quantum principal component analysis~\cite{li2021resonant, xin2021experimental, tang2021quantum, huang2022quantum}. 
Since the process of DME is independent of the particular states $\rho$ and $\sigma$, and the number of copies of the program state $\sigma$ for DME to achieve a desired precision level does not explicitly depend on its dimension, this gives an exponential advantage for quantum principal component analysis. 
Also, DME offers a sample-based strategy for Hamiltonian simulation~\cite{Kimmel2017}.
Specifically, if the Hamiltonian to be simulated can be encoded in a program state and multiple copies of it are available, one can simulate the dynamics induced by the corresponding Hamiltonian by using copies of the program state. 
Furthermore, DME can be employed in other tasks, such as block-encoding a density matrix when given access to its samples (as noted in~\cite{GLM+22} and investigated in~\cite{GP22, WZ23, WZ24}), quantum machine learning~\cite{lloyd2014quantum, rebentrost2014quantum},
and revealing entanglement spectra~\cite{pichler2016measurement}.
Along with these applications, there has been experimental progress in implementing DME~\cite{kjaergaard2022demonstration}.

The essence of DME is that it becomes more accurate as the number of program states (i.e., sample number) increases.
Here, a crucial question is to determine how many copies of the program state are required to realize DME up to the desired imprecision level, 
which is also referred to as the sample complexity of DME~\cite{Kimmel2017} (see also~\cite{arunachalam2018optimal, canonne2022topics} for the notion of sample complexity more generally).
Refs.~\cite{lloyd2014quantum, Kimmel2017} claimed that in the asymptotic limit, the sample complexity of DME is given by $O(t^2/\varepsilon)$, in order to achieve a desired evolution time $t \geq 0$ and imprecision level $\varepsilon \in [0,1]$.
Ref.~\cite{Kimmel2017} also showed that this scaling behavior of the sample complexity is optimal, such that the sample complexity of an arbitrary protocol for sample-based Hamiltonian simulation cannot scale as $o(t^2/\varepsilon)$. 
However, as argued in Appendix~\ref{appendix: incomplete} of our paper, the previous proof from~\cite[Supplementary Information Section~A]{Kimmel2017}, for the upper bound of DME sample complexity, is incomplete. That is, the proof given in~\cite[Supplementary Information Section~A]{Kimmel2017} does not give a strict bound for an arbitrary evolution time. 

In this paper, we rigorously establish the sample complexity of DME, in terms of the desired imprecision level $\varepsilon$ and evolution time $t$. 
More precisely, we show that the sample complexity of DME is no larger than $4t^2/\varepsilon$ (Theorem~\ref{Thm: error bound}), where $t$ is the desired evolution time and $\varepsilon$ is the desired imprecision level, as quantified by the normalized diamond distance. A sample complexity analysis for DME was previously presented in~\cite{wei2024simulating}, in the context of a more general task called Hermitian-preserving map exponentiation, which aims to simulate the exponential of the output of a Hermitian-preserving map. As discussed in more detail in our paper, our results offer a slight improvement over the previous bounds given in~\cite{wei2024simulating}.
We also investigate a fundamental lower bound on the sample complexity of an arbitrary protocol for sample-based Hamiltonian simulation by employing the notion of zero-error query complexity and appealing to prior results of~\cite{acin2001statistical} (see Theorem~\ref{Thm: optimality} and Lemma~\ref{lem: lower bound of sample complexity}).
This fundamental lower bound shows that the sample complexity of DME is optimal even in the non-asymptotic regime, up to a multiplicative constant factor.

\subsection{Wave matrix Lindbladization}
A natural analogue to DME, for the task of sample-based Lindbladian simulation, is known as Wave Matrix Lindbladization (WML)~\cite{patel2023wave1, patel2023wave2}. In WML, one is given access to an input quantum state $\rho$ and multiple copies of a program state $|\psi_L\rangle$ that encodes a Lindblad operator $L$ as follows: 
\begin{equation}
    \label{eq:WML-program-state}
    \ket{\psi_L} \coloneqq  (L \otimes I) \ket{\Gamma},
\end{equation}
where $|\Gamma\rangle \coloneqq \sum_{j} |j\rangle \otimes |j\rangle$ denotes the (unnormalized) maximally entangled vector. The goal is to  implement the Lindbladian evolution $e^{ \mathcal{L} t}$ approximately for a desired evolution time $t \geq 0$, where the Lindbladian $\mathcal{L}$ is defined as follows:
\begin{equation}
\label{eq:lindblad_dynamics}
\mathcal{L}(\rho) \coloneqq L \rho L^\dagger - \frac{1}{2} \left\{L^\dagger L, \rho \right\}.
\end{equation}
While more general Lindbladians include both Hamiltonian terms and multiple dissipative terms, our focus in this paper is on the single-Lindblad-operator setting, which serves as the foundation for understanding more general algorithms presented in~\cite{patel2023wave2}.

Similar to DME, the sample complexity of WML is defined as the number of copies of the program state $|\psi_L\rangle$ required to realize the target quantum channel $e^{ \mathcal{L} t}$ up to a desired imprecision level $\varepsilon \in [0, 1]$. Ref.~\cite{patel2023wave1} established that, in the asymptotic limit and under the assumption that $L$ is a local operator acting on a constant number of qubits, the sample complexity of WML scales as $O(t^2/\varepsilon)$. Although this locality assumption is valid for most physically relevant systems, a comprehensive analysis for arbitrary (not necessarily local) Lindbladians is crucial for a more complete theoretical understanding. Moreover, Ref.~\cite{patel2023wave1} does not address fundamental lower bounds on the sample complexity of WML, leaving open the question of its optimality.

In this paper, we rigorously establish the sample complexity of WML for general single-operator Lindbladians. We show that the sample complexity is bounded from above by $3t^2d^2/\varepsilon$, where $d$ is the dimension of the system on which the Lindblad operator acts nontrivially. Furthermore, we derive a fundamental lower bound on the sample complexity of any protocol for sample-based Lindbladian simulation, using techniques analogous to those developed for the DME setting. These results imply that the WML protocol is sample optimal up to constant factors in the case when the Lindblad acts nontrivially on a constant-sized system.

Using the lower bounds for sample-based Hamiltonian and sample-based local Lindbladian simulation, we also derive a lower bound for general Lindbladians that include both types of terms. This shows that the sampling-based WML algorithm from Ref.~\cite{patel2023wave2} for simulating general Lindbladians is sample-optimal up to constant factors for the case when the Lindblad operators act locally.

\subsection{Paper organization}
Our paper is organized as follows. 
In Section~\ref{sec:background}, we establish notation, review sample-based Hamiltonian simulation and DME, and also review sample-based Lindbladian simulation and WML. In Section~\ref{sec:sample-complex-def}, we provide a formal definition of the sample complexity of sample-based Hamiltonian and Lindbladian simulation.
In Section~\ref{sec:upper-bound-DME}, we establish an upper bound on the sample complexity of DME, and thus on the sample complexity of sample-based Hamiltonian simulation. In Section~\ref{sec:lower-bnd-samp-complex-SBHS}, we establish a lower bound on the sample complexity of sample-based Hamiltonian simulation.
In Section~\ref{sec:upper-bound-WML}, we establish an upper bound on the sample complexity of WML, and thus on the sample complexity of sample-based Lindbladian simulation. In Section~\ref{sec:lower-bound-SBLS}, we establish a lower bound on the sample complexity of sample-based Lindbladian simulation.
Table~\ref{tab:summary-of-results} summarizes some of the main results of our paper.
Finally, in Section~\ref{sec:conclusion}, we provide some concluding remarks and directions for future research.

\renewcommand{\arraystretch}{1.5}

\begin{table}
    \centering
    \begin{tabular}{|c|c|c|}
    \hline\hline
         &  Hamiltonian simulation & Lindbladian simulation\\
         \hline\hline
         Target dynamics & $\rho \to e^{-i\sigma t} \rho e^{i \sigma t}$ & $\rho \to e^{\mathcal{L}t}(\rho)$\\
         \hline
         Algorithm & DME &  WML\\
         \hline
         Sample complexity upper bound & $4\left(\frac{t^2}{\varepsilon}\right)  $ (Corollary~\ref{cor:samp-comp-upper-bnd}) & $3d^2 \left(\frac{t^2}{\varepsilon} \right)$ (Corollary~\ref{cor:samp-comp-upper-bnd-WML})\\
         \hline
         Sample complexity lower bound & $\frac{32}{1000} \left(\frac{t^2}{\varepsilon} \right)$ (Theorem~\ref{Thm: optimality}) & $10^{-4} \left(\frac{t^2}{\varepsilon} \right)$ (Theorem~\ref{thm:lower_bound_Lindbladian}) \\
         \hline\hline
    \end{tabular}
    \caption{This table highlights some of the main results of our paper for sample-based Hamiltonian and Lindbladian simulation. DME stands for density matrix exponentiation \cite{lloyd2014quantum} and WML for wave matrix Lindbladization \cite{patel2023wave1}. The table entries list our sample complexity results in terms of the simulation time $t\geq 0$ and the desired error $\varepsilon \in (0,1]$, the latter defined with respect to normalized diamond distance. The variable $d$ is the dimension of the system on which the Lindblad operator acts non-trivially. The table indicates that, up to constant factors, DME is optimal for sample-based Hamiltonian simulation, and WML is optimal for sample-based Lindbladian simulation if $d$ is constant.}
    \label{tab:summary-of-results}
\end{table}

\section{Background}

\label{sec:background}

\subsection{Notation}

This section introduces some notation that will be used throughout the rest of our paper. Let $\mathcal{H}_{S}$ denote the Hilbert space corresponding to a quantum system $S$. 
Let $\mathcal{D}(\mathcal{H}_{S})$ denote the set of quantum states (density operators) acting on the Hilbert space $\mathcal{H}_{S}$. 
We sometimes use the notation $\mathcal{H}^{d}$ to denote a Hilbert space in terms of its dimension $d$ and the notation $\mathcal{D}(\mathcal{H}^{d})$ to denote the set of $d$-dimensional quantum states. 

Let $\operatorname{Tr}[X]$ denote the trace of a matrix $X$, and let $X^{\dag}$ denote the Hermitian conjugate of the matrix $X$. 
For $p \in [1,\infty)$, the Schatten $p$-norm of the matrix $X$ is defined as
\begin{align}\label{eq: schatten-p-norm}
    \left\|X\right\|_{p} \coloneq \left( \operatorname{Tr}\!\left[\left(XX^{\dag}\right)^{\frac{p}{2}} \right]\right)^{\frac{1}{p}}.
\end{align}
Throughout this work, we use $\left\| \cdot \right\|_{1}$, which is referred to as the trace norm, and $\left\| \cdot\right \|_{\infty}$, which is referred to as the operator norm.

To quantify the distance (or closeness) between two quantum states $\rho,\sigma \in \mathcal{D}(\mathcal{H}_{S})$, we use the normalized trace distance, which is equal to the normalized trace norm (i.e., $p = 1$ in~\eqref{eq: schatten-p-norm}) of their difference, such that 
\begin{align}
    \frac{1}{2}\left\| \rho - \sigma \right\|_{1} . 
\end{align}
This quantity is equal to the maximum difference in the probabilities that an arbitrary measurement operator can assign to these two states~\cite[Eq.~(9.22)]{nielsen2010quantum}. 
Note that the multiplicative factor $\frac{1}{2}$ above guarantees that $\frac{1}{2}\left\| \rho - \sigma \right\|_{1} \in [0,1]$ for any given quantum states $\rho$ and $\sigma$. 
Throughout this work, the normalization factor $\frac{1}{2}$ will sometimes be omitted when using the trace distance.
We also use the fidelity of quantum states, which quantifies the closeness between two quantum states.
Specifically, the fidelity $F(\rho, \sigma)$ between two quantum states $\rho,\sigma$ is defined as~\cite{uhlmann1976transition}
\begin{align}\label{eq: fidelity definition}
    F(\rho, \sigma) \coloneqq \left\| \sqrt{\rho} \sqrt{\sigma} \right\|_1^2 = \operatorname{Tr}\!\left[\sqrt{\sqrt{\rho}\sigma\sqrt{\rho}}\right]^2.
\end{align}

To quantify the distance between two quantum channels (completely positive and trace-preserving maps), we use the normalized diamond distance, which is defined for two quantum channels $\mathcal{N}$ and $\mathcal{M}$ as~\cite{kitaev1997quantum}
\begin{align}
    \frac{1}{2}\left\| \mathcal{N} - \mathcal{M} \right\|_{\diamond}  \coloneqq \frac{1}{2}\sup_{\rho \in \mathcal{D}\left(\mathcal{H}_{R}\otimes \mathcal{H}_{S}\right)}\left\|(\mathcal{I}_{R}\otimes\mathcal{N})(\rho) - (\mathcal{I}_{R}\otimes\mathcal{M})(\rho) \right\|_{1} ,
\end{align}
where $R$ denotes a reference system with an arbitrarily large Hilbert space dimension and $\mathcal{I}_{R}$ denotes the identity channel acting on the reference Hilbert space $\mathcal{H}_{R}$. 
An important point is that, by definition, the dimension of $R$ can be arbitrarily large, but one can place a bound on the dimension of $R$ equal to the dimension of $S$~\cite[Theorem~9.1.1]{wilde2017QuantumInformationTheory}.
Also, note that the normalized diamond distance between two quantum channels involves an optimization of the normalized trace distance between two arbitrary quantum states over $\mathcal{D}(\mathcal{H}_{R}\otimes \mathcal{H}_{S})$, and accordingly, it is guaranteed that $\frac{1}{2}\left\| \mathcal{N} - \mathcal{M} \right\|_{\diamond} \in [0,1]$, due to the normalization factor~$\frac{1}{2}$.

For a quantum state $\rho \in \mathcal{D}(\mathcal{H}_{S_1} \otimes \mathcal{H}_{S_2})$ of systems $S_1$ and $S_2$, we denote the partial trace of $\rho$ over the Hilbert space $\mathcal{H}_{S_2}$ by $\operatorname{Tr}_{S_2}[\rho]$.
Let $\mathbb{I}_{S}\coloneqq \sum_{i}|i\rangle\!\langle i|_{S}$  denote the identity operator acting on system $\mathcal{H}_S$, where $\left\{\ket{i}\right\}_i$ is an orthonormal basis.
We further define the swap operator $\operatorname{SWAP}$ between two systems $\mathcal{H}_{S_1}$ and $\mathcal{H}_{S_2}$ in the following way: 
\begin{align}
    \operatorname{SWAP}_{S_1S_2}\coloneqq\sum_{i,j}|i\rangle\!\langle j|_{S_1} \otimes|j\rangle\!\langle i|_{S_2}.
    \label{eq:def-swap}
\end{align}

\subsection{Review of sample-based Hamiltonian simulation and density matrix exponentiation}

In this section, we review the task of sample-based Hamiltonian simulation and the DME algorithm for achieving this task~\cite{lloyd2014quantum,Kimmel2017}.
The aim of sample-based Hamiltonian simulation is to perform the following task:
on input one copy of an unknown quantum state $\rho$ and $n$ copies of a program  state $\sigma$, implement the unitary operation $e^{-i\sigma t}$ on $\rho$ for evolution time $t \geq 0$ to within imprecision level $\varepsilon \in [0,1]$~\cite{lloyd2014quantum}.
In short, the task is to perform a fixed quantum channel $\mathcal{P}^{(n)}$ (independent  of the program state $\sigma$) such that the following inequality holds:
\begin{align}
    \frac{1}{2} \left \| \mathcal{P}^{(n)}\circ  \mathcal{A}_{\sigma^{\otimes n}}  - \mathcal{U}_{\sigma, t} \right \|_{\diamond} \leq \varepsilon,
\end{align}
where $\mathcal{A}_{\sigma^{\otimes n}}$ denotes the channel that appends the state $\sigma^{\otimes n}$ to the input (i.e., $\mathcal{A}_{\sigma^{\otimes n}}(\rho) \coloneqq \rho \otimes \sigma^{\otimes n}$) and $\mathcal{U}_{\sigma, t}(\cdot)$ denotes  the target unitary channel to be approximated by a sample-based Hamiltonian simulation algorithm, which is the ideal evolution according to the quantum state $\sigma$ for the evolution time $t\geq0$. 
More formally,
\begin{equation}
\label{eq: ideal evolution by t}
\mathcal{U}_{\sigma, t}(\rho)\coloneqq e^{-i\sigma t}\rho e^{i\sigma t}.
\end{equation}
We denote the number of copies of the input program state $\sigma$ by $n$, where $n\in\mathbb{N}$ and satisfies $n>t$.

One can divide the evolution time $t$ by the number of copies $n$, and we denote the resulting fraction as a unit time step $\Delta$; i.e.,  $\Delta\coloneqq\frac{t}{n}$.
By means of the Hadamard lemma (see, e.g.,~\cite[Lemma~19]{Kimmel2017}), the ideal evolution of the state $\rho$ by the Hamiltonian $\sigma$ for a time step $\Delta$ can be expressed as a series in $\Delta$: 
\begin{align}
\mathcal{U}_{\sigma, \Delta}(\rho) & =e^{-i\sigma\Delta}\rho e^{i\sigma\Delta} \label{eq:DME-ideal-one-step-1}\\
& = \rho -i\Delta[\sigma,\rho] -\frac{1}{2!}\Delta^2[\sigma, [\sigma, \rho]] + \cdots 
\\
& = \sum_{j=0}^{\infty} \frac{(-i\Delta)^j}{j!} [\sigma,\rho]_{j},
\label{eq:DME-ideal-one-step}
\end{align}
where we define the nested commutator similarly to~\cite{Kimmel2017} as 
\begin{equation}
[X,Y]_{k}  \coloneqq \underbrace{[X,\dotsb[X,[X}_{k \text { times }}, Y]] \dotsb], \quad \text{while}
\quad [X,Y]_{0}  \coloneqq Y.
\label{eq:def-nested-comm}
\end{equation}
Clearly, $n$ repetitions of the ideal evolution $\mathcal{U}_{\sigma, \Delta}$ lead to the  desired ideal evolution $\mathcal{U}_{\sigma, t}$:
\begin{equation}
\mathcal{U}_{\sigma, t}(\rho)=\mathcal{U}_{\sigma, \Delta}^{n}(\rho).
\label{eq:n-times-ideal}
\end{equation}

The DME algorithm approximates the ideal time evolution $\mathcal{U}_{\sigma, \Delta}(\rho)$ for each time step $\Delta$ in~\eqref{eq:DME-ideal-one-step} by utilizing a single copy of the program state $\sigma$. More specifically, given a quantum state~$\rho$ in the system $S_1$ and a single copy of the program state $\sigma$ in the ancillary system~$S_2$, DME applies the swap Hamiltonian between systems $S_1$ and $S_2$ for the unit time step~$\Delta$, and then discards the ancillary system $S_2$.
In short, for each step, DME realizes the following quantum channel:
\begin{equation}
\widetilde{\mathcal{U}}_{\sigma, \Delta}(\rho_{S_1})\coloneqq\operatorname{Tr}
_{S_2}\!\left[e^{-i\Delta \operatorname{SWAP}_{S_1S_2}}(\rho_{S_1}\otimes\sigma_{S_2})e^{i\Delta \operatorname{SWAP}_{S_1S_2}}\right],
\label{eq:DME-actual-one-step}
\end{equation}
where $\operatorname{SWAP}_{S_1S_2}$ is defined in~\eqref{eq:def-swap}.
For a sufficiently small time step $\Delta$, the quantum channel generated by DME in~\eqref{eq:DME-actual-one-step} is close to the ideal evolution in~\eqref{eq:DME-ideal-one-step}. 
More specifically, as argued in Ref.~\cite{lloyd2014quantum}, the following holds:
\begin{align}\label{eq:DME-actual-one-step-expansion}
    \operatorname{Tr}
_{S_2}\!\left[e^{-i\Delta \operatorname{SWAP}_{S_1S_2}}(\rho_{S_1}\otimes\sigma_{S_2})e^{i\Delta \operatorname{SWAP}_{S_1S_2}}\right]
    = \rho -i\Delta[\sigma, \rho] + O(\Delta^2),
\end{align}
where we neglected the system label $S_1$ on the right-hand side of~\eqref{eq:DME-actual-one-step-expansion} for the sake of brevity. 
Since the right-hand side of~\eqref{eq:DME-actual-one-step-expansion} is equal to the right-hand side of~\eqref{eq:DME-ideal-one-step} up to the first order in $\Delta$, the error of DME for each time step $\Delta$ is asymptotically bounded by the second order in $\Delta$, i.e.,
\begin{equation}\label{eq: error induced for each time step}
\mathcal{U}_{\sigma, \Delta}(\rho) - \widetilde{\mathcal{U}}_{\sigma, \Delta}(\rho) = O(\Delta^2) .
\end{equation}
Hence, for each step over $n$ steps, DME approximates the ideal evolution using a single copy of $\sigma$, up to the imprecision level $O(\Delta^2)$.  
By repeating this process $n$ times (i.e., implementing $\widetilde{\mathcal{U}}_{\sigma, \Delta}^{n}(\rho)$) and thus consuming $n$ program states, DME finally approximates the ideal evolution $\mathcal{U}_{\sigma, t}(\rho)$.

It was argued in~\cite{lloyd2014quantum}, by appealing to the Trotter--Suzuki theory of Hamiltonian simulation, that the total error scales as $O(n \Delta^2) = O(t^2 / n)$, so that $n = O(t^2/\varepsilon)$ copies of the program state are needed to have a total simulation error no larger than $\varepsilon$. The same conclusion was reached in~\cite{Kimmel2017} by means of a different argument. In fact, Ref.~\cite{Kimmel2017} claimed that for every input program state $\sigma$, the diamond distance between $\mathcal{U}_{\sigma, t}$ and $\widetilde{\mathcal{U}}_{\sigma, \Delta}^{n}$ is asymptotically bounded from above by $O(t^2/n)$. 
This indicates that $O(t^2/\varepsilon)$  samples are required to achieve the imprecision level $\varepsilon$.
However, as argued in Appendix~\ref{appendix: incomplete}, the proof of this previous sample complexity bound is incomplete, in the sense that it does not give a strict imprecision bound for an arbitrary evolution time $t$.

\subsection{Review of sample-based Lindbladian simulation and wave matrix Lindbladization}

We now review the task of sample-based Lindbladian simulation and the WML algorithm~\cite{patel2023wave1, patel2023wave2}. The ideal evolution of a state $\rho$ for time $t\geq0$ according to the Lindbladian in~\eqref{eq:lindblad_dynamics} is as follows:
\begin{equation}
    \label{eq: ideal lindblad evolution by t}
    e^{ \mathcal{L} t}(\rho) = \sum_{k=0}^{\infty}\mathcal{L}^{k}(\rho)\frac{t^{k}}{k!}.
\end{equation}
Similar to sample-based Hamiltonian simulation, the task of Lindbladian simulation is to perform a fixed quantum channel ${\cal P}^{(n)}$ on one copy of an arbitrary input quantum state $\rho$ and $n$ copies of a program state $\psi_L$ that implements the Lindbladian evolution $e^{\mathcal{L} t}$, within imprecision level $\varepsilon$:
\begin{align}
    \frac{1}{2} \left \| \mathcal{P}^{(n)} \circ {\cal A}_{\psi_L^{\otimes n}}  - e^{{\cal L} t} \right \|_{\diamond} \leq \varepsilon.
\end{align}

By dividing $t$ into $n$ intervals, each of duration $\Delta\coloneqq\frac{t}{n}$, we have that $e^{ \mathcal{L} t}=\left(  e^{\mathcal{L}\Delta}\right)  ^{n}$. Simulation of one time step according to the wave matrix Lindbladization (WML) algorithm is
equivalent to the following quantum channel:
\begin{equation}
\widetilde{e^{\mathcal{L}\Delta}}
(\rho)\coloneqq\operatorname{Tr}_{23}[e^{\mathcal{M}\Delta}  (\rho\otimes\psi_{L})],\label{eq:WML-actual-one-step}%
\end{equation}
where $\psi_{L} = \ket{\psi_L} \!\bra{\psi_L}$ is the program state defined to be a bipartite state in a $d\times
d$-dimensional Hilbert space:
\begin{align}
|\psi_{L}\rangle & \coloneqq \left(  L\otimes I\right)  |\Gamma\rangle,\\
|\Gamma\rangle & \coloneqq \sum_{i=1}^{d}|i\rangle|i\rangle,
\end{align}
and ${\cal M}$ is a fixed Lindbladian independent of the program state, given as
\begin{align}
\mathcal{M}(\omega)  & \coloneqq M\omega M^{\dag}-\frac{1}{2}\left\{  M^{\dag}%
M,\omega\right\}  , \label{eq:WML-M-op} \\
M  & \coloneqq \frac{1}{\sqrt{d}} \left(  I_{1}\otimes|\Gamma\rangle\!\langle\Gamma|_{23}\right)  \left(
\mathrm{SWAP}_{12}\otimes I_{3}\right).
\end{align}
Without loss of generality, we assume that $\left\| L\right\|_2=1$, so that  $\ket{\psi_L}$ is normalized. 

In Refs.~\cite {patel2023wave1, patel2023wave2}, it was argued that the difference between the ideal and approximate dynamics for a short time step $\Delta$ is bounded as
\begin{equation}
    e^{ \mathcal{L} \Delta}(\rho) - \widetilde{e^{\mathcal{L}\Delta}}(\rho) = O(\Delta^2),
\end{equation}
implying that the WML algorithm simulates the Lindbladian dynamics up to error $\Delta^2$. By repeating this process $n$ times, the total error scales as $O(n\Delta^2) = O(t^2/n)$, so that $n = O(t^2/\varepsilon)$ copies of the program state are needed to have a total simulation error no larger than $\varepsilon$, which is the same as that of Hamiltonian simulation. However, since the proof in \cite{patel2023wave1} followed the same approach as in \cite{Kimmel2017}, it suffers from similar issues as discussed in Appendix~\ref{appendix: incomplete}. In this paper, we clarify this issue and identify a precise upper bound on the sample complexity of WML. \\

\section{Definition of sample complexity of sample-based Hamiltonian and Lindbladian simulation}

\label{sec:sample-complex-def}

The main goal of our paper is to examine the sample complexity of DME \red{and WML, which lead to sample-based Hamiltonian and Lindbladian simulations, respectively.}
More formally, we investigate the precise number of program states, $n$, required to obtain a desired imprecision level $\varepsilon$ in the diamond distance for an evolution time $t$.
As such, one of main goals for DME \red{and WML} is to find the minimum value of $n$ such that the following error bound holds for arbitrarily given $t \geq 0$, $\varepsilon \in [0,1]$, and $\sigma\in \mathcal{D}(\mathcal{H})$:
\begin{equation}
\label{eq:channel_error_definition}
\frac{1}{2} \left\Vert \mathcal{E}_{\sigma, t}-\widetilde{\mathcal{E}}_{\sigma, \Delta}^{n}\right\Vert
_{\diamond}\leq \varepsilon.
\end{equation}
Here, the ideal channel $\mathcal{E}_{\sigma, t}$ corresponds to ${\cal U}_{\sigma, t}$ in~\eqref{eq: ideal evolution by t} for DME and $e^{t {\cal L}}$ with $\sigma = \ket{\psi_L} \bra{\psi_L}$ in~\eqref{eq: ideal lindblad evolution by t} for WML, respectively. Similarly, the approximated channel $\widetilde{\mathcal{E}}_{\sigma, \Delta}$ for a small time interval $\Delta$ corresponds to $\widetilde {\cal U}_{\sigma, \Delta}$ in~\eqref{eq:DME-actual-one-step} for DME and $\widetilde{e^{\mathcal{L}\Delta}}$ in~\eqref{eq:WML-actual-one-step} for WML, respectively.

Let us now formally define the sample complexity, more generally, for sample-based Hamiltonian and Lindbladian simulation, which includes the quantum channel induced by the DME and WML algorithms.
Here, the sample complexity is the minimum number of program states needed to approximate the unitary operation up to a desired imprecision level $\varepsilon$, where now the minimization is over all possible quantum channels.  

\begin{definition}[Sample complexity of
Hamiltonian and Lindbladian simulation]
\label{def: sample complexity}
The sample complexity of sample-based Hamiltonian and Lindbladian simulation is denoted by $n_{d}^{\ast}(t,\varepsilon)$ and is defined as the minimum number of program states, the latter denoted by $\kappa$, required to realize a channel that is $\varepsilon$-distinguishable in normalized diamond distance from the ideal channel
$\mathcal{E}_{\kappa,t}$, as defined in
\eqref{eq:channel_error_definition},
for an arbitrary program
state $\kappa$ of dimension $d$.
Formally, the sample complexity $n_{d}^{\ast}(t,\varepsilon)$ is defined as
\begin{align}
n_{d}^{\ast}(t,\varepsilon) & \coloneqq \inf_{\mathcal{P}^{(n)}\in\rm{CPTP}
}\left\{  n\in\mathbb{N}:\frac{1}{2}\left\Vert \mathcal{P}^{(n)}
\circ\mathcal{A}_{\kappa^{\otimes n}}-\mathcal{E}_{\kappa,t}\right\Vert
_{\diamond}\leq\varepsilon\ \ \forall\kappa\in\mathcal{D}(\mathcal{H}^{d})\right\}
\label{eq:sample-complexity-ham-sim}\\
& =  \inf_{\mathcal{P}^{(n)}\in\rm{CPTP}
}\left\{  n\in\mathbb{N}:\sup_{\kappa\in\mathcal{D}(\mathcal{H}^{d})} \frac{1}{2}\left\Vert \mathcal{P}^{(n)}
\circ\mathcal{A}_{\kappa^{\otimes n}}-\mathcal{E}_{\kappa,t}\right\Vert
_{\diamond}\leq\varepsilon\right\},
\end{align}
where $\mathcal{D}(\mathcal{H}^{d})$ denotes the set of quantum states of dimension $d$ and the appending channel $\mathcal{A}_{\kappa^{\otimes n}}$ acting on an arbitrary input state~$\zeta$ is defined as
\begin{equation}
\mathcal{A}_{\kappa^{\otimes n}}(\zeta)\coloneqq \zeta\otimes\kappa^{\otimes
n},
\end{equation}
where $\kappa$ is a quantum state.
\end{definition}

It is worth emphasizing that the quantum channel $\mathcal{P}^{(n)}$ in the right-hand side of~\eqref{eq:sample-complexity-ham-sim} includes the DME and WML operations. Therefore, an upper bound on the sample complexity of DME and WML, with a program state being of dimension $d$, certainly provides an upper bound on $n_{d}^{\ast}$.

\section{Upper bound on the sample complexity of density matrix exponentiation}

\label{sec:upper-bound-DME}

In this section, we examine the sample complexity of DME (i.e., the number of program states needed to achieve a desired imprecision level $\varepsilon$ and evolution time $t$). 
Here, we identify an upper bound on the imprecision induced by the whole process of DME, in terms of  number of samples, $n$, and evolution time $t$. 
Our first result is stated in the following theorem:

\begin{theorem}\label{Thm: error bound}
Let $t\geq0$, let $n\in\mathbb{N}$ be such that $n>t$, and let $\Delta\coloneqq\frac{t}{n}$. 
For every quantum state $\sigma$, the error of density matrix exponentiation satisfies the following bound:
\begin{equation}
\frac{1}{2} \left\Vert \mathcal{U}_{\sigma, t}-\widetilde{\mathcal{U}}_{\sigma, \Delta}^{n}\right\Vert
_{\diamond}\leq\frac{4t^{2}}{n}, \label{eq:DME-err-bound}
\end{equation}
where $\mathcal{U}_{\sigma, t}$ is defined in~\eqref{eq: ideal evolution by t} and $\widetilde{\mathcal{U}}_{\sigma, \Delta}$ in~\eqref{eq:DME-actual-one-step}. 
\end{theorem}

Here we emphasize that an analysis similar to that presented in the proof of Theorem~\ref{Thm: error bound} was already established in~\cite{wei2024simulating} for a more general quantum algorithm beyond DME.
When specializing to the DME setup, their sample complexity bound yields essentially the same constant prefactor as given in~\eqref{eq:DME-err-bound}.
Nevertheless, we have included this theorem for completeness, not only to make our analysis more self-contained, but also to explicitly clarify the gap between the sample complexity upper and lower bounds. We also note that, in comparison to~\cite{wei2024simulating}, our proof employs the slightly relaxed constraint $n>t$, which was previously given as $0.8n>t$ in~\cite{wei2024simulating}.

As a consequence of Theorem~\ref{Thm: error bound}, we now obtain an upper bound on the sample complexity of DME as follows.

\begin{corollary}[Upper bound on the sample complexity of DME]
\label{cor:samp-comp-upper-bnd}
Theorem~\ref{Thm: error bound} implies  the following dimension-free upper bound on the sample complexity of sample-based Hamiltonian simulation:
\begin{equation}
    n_{d}^{\ast}(t,\varepsilon) \leq \frac{4t^2}{\varepsilon}.
\end{equation}
\end{corollary}

Hence, Corollary~\ref{cor:samp-comp-upper-bnd} resolves the incompleteness of the claim from Ref.~\cite{Kimmel2017} and validates the previously claimed sample complexity bound $O(t^2/\varepsilon)$ from Ref.~\cite{lloyd2014quantum}.

\medskip 

\begin{proof}[Proof of Theorem~\ref{Thm: error bound}]
Applying~\eqref{eq:n-times-ideal}, we find that
\begin{align}
\frac{1}{2} \left\Vert \mathcal{U}_{\sigma, t}-\widetilde{\mathcal{U}}_{\sigma, \Delta}^{n}\right\Vert
_{\diamond}  &  =\frac{1}{2} \left\Vert \mathcal{U}_{\sigma, \Delta}^{n}-\widetilde{\mathcal{U}
}_{\sigma, \Delta}^{n}\right\Vert _{\diamond}\\
&  \leq  n \cdot \frac{1}{2} \left\Vert \mathcal{U}_{\sigma, \Delta}-\widetilde{\mathcal{U}}_{\sigma,\Delta
}\right\Vert _{\diamond}\\
&  \leq n(4\Delta^{2}),  \label{eq: final bound}
\end{align}
where we have inductively applied the subadditivity of the diamond distance (see Lemma~\ref{lem:subadd-DD}) to obtain the
first inequality. 
The second inequality follows from Lemma~\ref{lem:single-step-DME-bound}, which provides a precise upper bound on the error induced by DME, for each time step of size $\Delta$. 
The right-hand side of the second inequality then implies~\eqref{eq:DME-err-bound} after substituting $\Delta=\frac{t}{n}$.
\end{proof}

\medskip

In the rest of the section, we provide a precise error bound induced when simulating a single step of DME for the unit time~$\Delta$. 
\begin{lemma}
\label{lem:single-step-DME-bound}
Let $\sigma$ be an arbitary quantum state, and suppose that $\Delta \in [0,1)$. 
For the quantum channels $\mathcal{U}_{\sigma, \Delta}$ and $\widetilde{\mathcal{U}}_{\sigma, \Delta}$ defined in~\eqref{eq:DME-ideal-one-step-1} and~\eqref{eq:DME-actual-one-step},
respectively, the following inequality holds:
\begin{equation}
\frac{1}{2} \left\Vert \mathcal{U}_{\sigma, \Delta}-\widetilde{\mathcal{U}}_{\sigma, \Delta}\right\Vert
_{\diamond}\leq 4\Delta^{2}. \label{eq:single-step-DME-bnd}
\end{equation}

\end{lemma}

\begin{proof}
Let $\rho \in \mathcal{D}(\mathcal{H}_{R}\otimes\mathcal{H}_{S_1})$ be an unknown bipartite quantum state over the joint system $RS_1$, where $R$ is an arbitrary reference system. 
Also, let $\sigma \in \mathcal{D}(\mathcal{H}_{S_2})$ be a program quantum state over the system $S_2$. 
Then, from~\eqref{eq:DME-actual-one-step}, we have
\begin{align}
&(\mathcal{I}_{R}\otimes\widetilde{\mathcal{U}}_{\sigma, \Delta})(\rho_{RS_1})  \nonumber \\
&  =\operatorname{Tr}_{S_2}\!\left[(\mathbb{I}_{R} \otimes e^{-i\Delta \operatorname{SWAP}_{S_1S_2}})(\rho_{RS_1}\otimes\sigma_{S_2})(\mathbb{I}_{R} \otimes e^{i\Delta \operatorname{SWAP}_{S_1S_2}})\right]\\
& =\operatorname{Tr}_{S_2}\!\left[
\begin{array}
[c]{c}\left(\mathbb{I}_{R} \otimes (\cos\Delta\cdot\mathbb{I}_{S_1S_2}-i\sin\Delta\cdot\operatorname{SWAP}_{S_1S_2})\right)(\rho_{RS_1}\otimes\sigma_{S_2}) \\ \times \left(\mathbb{I}_{R} \otimes (\cos\Delta\cdot\mathbb{I}_{S_1S_2}+i\sin\Delta\cdot\operatorname{SWAP}_{S_1S_2}\right)\end{array}\right] \\
&  = \left(  \cos^{2}\Delta\right)\operatorname{Tr}_{S_2}\!\left[ (\rho_{RS_1}\otimes\sigma_{S_2}) \right]  +\left(  \sin^{2}\Delta\right)  \operatorname{Tr}_{S_2}\!\left[ \operatorname{SWAP}_{S_1S_2} (\rho_{RS_1}\otimes\sigma_{S_2}) \operatorname{SWAP}_{S_1S_2} \right]  \nonumber \\
& \qquad-i\left(  \sin\Delta\cos\Delta\right)\operatorname{Tr}_{S_2}\!\left[\operatorname{SWAP}_{S_1S_2} (\rho_{RS_1}\otimes\sigma_{S_2}) - (\rho_{RS_1}\otimes\sigma_{S_2}) \operatorname{SWAP}_{S_1S_2} \right] \\
&  =\left(  \cos^{2}\Delta\right)  \rho_{RS_1} -i\left(  \sin\Delta\cos\Delta\right)
\left[ \left(\mathbb{I}_R \otimes \sigma_{S_1}\right),\rho_{RS_1}\right]  +\left(  \sin^{2}\Delta\right) \left(\operatorname{Tr}_{S_1}\!\left[\rho_{RS_1}\right] \otimes \sigma_{S_1}\right) \\
&  =\left(  1-\sin^{2}\Delta\right)  \rho_{RS_1} -i\frac{\sin2\Delta}{2}
\left[ \left(\mathbb{I}_R \otimes \sigma_{S_1}\right),\rho_{RS_1}\right]  +\left(  \sin^{2}\Delta\right) \left(\operatorname{Tr}_{S_1}\!\left[\rho_{RS_1}\right] \otimes \sigma_{S_1}\right) \\
&  =\rho_{RS_1}  -i\frac{\sin2\Delta}{2}
\left[ \left(\mathbb{I}_R \otimes \sigma_{S_1}\right),\rho_{RS_1}\right]  +\left(  \sin
^{2}\Delta\right)  \left(  \operatorname{Tr}_{S_1}\!\left[\rho_{RS_1}\right] \otimes \sigma_{S_1} - \rho_{RS_1}\right),
\end{align}
where we used the fact that $e^{-i\Delta \operatorname{SWAP}_{S_1S_2}} = \cos\Delta\cdot\mathbb{I}_{S_1S_2} -i\sin\Delta\cdot\operatorname{SWAP}_{S_1S_2}$, given that the $\operatorname{SWAP}$ operator is self-inverse (i.e., $\operatorname{SWAP}^2 = \mathbb{I}$).
Also, from~\eqref{eq:DME-ideal-one-step} and using the Hadamard Lemma in Ref.~\cite[Lemma~A1 of Supplementary Information]{Kimmel2017}, we have
\begin{align}
& \left(\mathcal{I}_{R}\otimes\mathcal{U}_{\sigma, \Delta}\right)(\rho_{RS_1}) \notag \\
& = \left(\mathbb{I}_{R}\otimes e^{-i\sigma\Delta}\right) \rho_{RS_1} \left( \mathbb{I}_{R}\otimes e^{i\sigma
\Delta}\right) \\
&  =\rho_{RS_1} -i \Delta\left[\left(\mathbb{I}_{R}\otimes  \sigma_{S_1}\right) ,\rho_{RS_1} 
 \right] +\sum_{k=2}^{\infty}\frac{\left(
-i\right)  ^{k}\Delta^{k}}{k!}\left[  (\mathbb{I}_R \otimes \sigma_{S_1}),\rho_{RS_1} \right]_{k} ,
\end{align}
where we previously defined the nested commutator $\left[  (\mathbb{I}_R \otimes \sigma_{S_1}),\rho_{RS_1} \right]_{k}$ in~\eqref{eq:def-nested-comm}.
Then, the trace distance between $ \left(\mathcal{I}_{R}\otimes\mathcal{U}_{\sigma, \Delta}\right)(\rho_{RS_1})$ and $(\mathcal{I}_{R}\otimes\widetilde{\mathcal{U}}_{\sigma, \Delta})(\rho_{RS_1})$ for an arbitrary quantum state $\rho_{RS_1}$ can be bounded from above as
\begin{align}
&  \left\Vert \left(\mathcal{I}_{R}\otimes\mathcal{U}_{\sigma, \Delta}\right)(\rho_{RS_1}) -(\mathcal{I}_{R}\otimes\widetilde{\mathcal{U}}_{\sigma, \Delta})(\rho_{RS_1})\right\Vert _{1}\nonumber\\
&  =\left\Vert   \begin{array}
[c]{c} -i\left(  \Delta-\frac{\sin2\Delta
}{2}\right) \left[\left(\mathbb{I}_{R}\otimes  \sigma_{S_1}\right) ,\rho_{RS_1} 
 \right]   + \sum_{k=2}^{\infty}\frac{\left(
-i\right)  ^{k}\Delta^{k}}{k!}\left[ ( \mathbb{I}_R \otimes \sigma_{S_1}),\rho_{RS_1} \right]_{k} \\ -  \left(  \sin
^{2}\Delta\right)  \left(  \operatorname{Tr}_{S_1}\!\left[\rho_{RS_1}\right] \otimes \sigma_{S_1} - \rho_{RS_1}\right) \end{array} \right\Vert _{1}\\
&  \leq   \left( \Delta-\frac{\sin2\Delta}{2}\right) \left\Vert\left[\left(\mathbb{I}_{R}\otimes  \sigma_{S_1}\right) ,\rho_{RS_1} \right]  \right\Vert _{1}  +\sum_{k=2}^{\infty}\frac{\Delta^{k}}
{k!}\left\Vert \left[  (\mathbb{I}_R \otimes \sigma_{S_1}),\rho_{RS_1} \right]_{k} \right\Vert _{1} \nonumber \\
& \qquad + \left(\sin^{2}\Delta\right) \left\Vert \operatorname{Tr}_{S_1}\!\left[\rho_{RS_1}\right] \otimes \sigma_{S_1} - \rho_{RS_1} \right\Vert _{1}  \\
&  \leq 2\left\Vert \sigma\right\Vert _{\infty}\left(  \Delta-\frac{\sin2\Delta}{2}\right)  +\sum_{k=2}^{\infty}
\frac{\left(  2\Delta\left\Vert \sigma\right\Vert _{\infty}\right)  ^{k}}
{k!}+2\sin^{2}\Delta  \label{eq: 35}\\
&  =2\left\Vert \sigma\right\Vert _{\infty}\left(  \Delta-\frac{\sin2\Delta
}{2}\right)  +e^{2\Delta\left\Vert \sigma\right\Vert _{\infty}}-1-2\Delta
\left\Vert \sigma\right\Vert _{\infty}+2\sin^{2}\Delta\\
&  \leq\frac{4}{3}\left\Vert \sigma\right\Vert _{\infty}\Delta^{3}+\frac{9}
{2}\Delta^{2}\left\Vert \sigma\right\Vert _{\infty}^{2}+2\Delta^{2} \label{eq: 37} \\
&  \leq8\Delta^{2}, \label{eq: 38}
\end{align}
where we used in~\eqref{eq: 35} the fact that
\begin{equation}
\left\Vert \left[\left(\mathbb{I}_{R}\otimes  \sigma_{S_1}\right) ,\rho_{RS_1} \right] \right\Vert _{1}
\leq
2\left\Vert\mathbb{I}_{R}\otimes  \sigma_{S_1} \right\Vert _{\infty}\left\Vert \rho_{RS_1} \right\Vert _{1} = 2\left\Vert
\sigma\right\Vert _{\infty},
\end{equation}
and then iterated this $k$ times to get that
\begin{align}
\left\Vert \left[(\mathbb{I}_{R}\otimes  \sigma_{S_1} ),\rho_{RS_1} \right]_{k} \right\Vert _{1} &\leq
2\left\Vert \sigma\right\Vert _{\infty}\left\Vert \left[(\mathbb{I}_{R}\otimes  \sigma_{S_1}) ,\rho_{RS_1} \right]_{k-1} \right\Vert _{1} \\
&\leq \cdots \\
&\leq \left(
2\left\Vert \sigma\right\Vert _{\infty}\right)  ^{k}\left\Vert \rho_{RS_1}\right\Vert
_{1}\\ 
& =\left(  2\left\Vert \sigma\right\Vert _{\infty}\right)  ^{k}.
\end{align}
We also used the fact that 
\begin{align}
    \left\Vert \operatorname{Tr}_{S_1}\!\left[\rho_{RS_1}\right] \otimes \sigma_{S_1} - \rho_{RS_1} \right\Vert _{1} \le 2 ,
\end{align}
given that both $\operatorname{Tr}_{S_1}\!\left[\rho_{RS_1}\right] \otimes \sigma_{S_1}$ and $ \rho_{RS_1}$ are quantum states. 
We also used the following inequalities to establish~\eqref{eq: 37}:
\begin{align}
\Delta-\frac{\sin2\Delta}{2}  &  \leq\frac{2}{3}\Delta^{3},\\
e^{2\Delta\left\Vert \sigma\right\Vert _{\infty}}-1-2\Delta\left\Vert
\sigma\right\Vert _{\infty}  &  \leq\frac{9}{2}\Delta^{2}\left\Vert
\sigma\right\Vert _{\infty}^{2},\\
\sin^{2}\Delta &  \leq\Delta^{2},
\end{align}
given that $\Delta < 1 $ and $\left\Vert \sigma \right\Vert_{\infty} \leq 1$.

We conclude that the normalized diamond distance between $ \mathcal{U}_{\sigma, \Delta}$ and $\widetilde{\mathcal{U}}_{\sigma, \Delta}$ is bounded from above by $4\Delta^2$ because the inequality in~\eqref{eq: 38} holds for an arbitrary input state $\rho_{RS_1}$. This completes the proof of~\eqref{eq:single-step-DME-bnd}.
\end{proof}

\section{Lower bound on the sample complexity of sample-based Hamiltonian simulation}

\label{sec:lower-bnd-samp-complex-SBHS}

In this section, we derive a lower bound on the sample complexity of sample-based Hamiltonian simulation, which is consistent with the optimality result from Ref.~\cite{Kimmel2017}. 
Using the definition of sample complexity given in Definition~\ref{def: sample complexity}, the main result of this section is as follows:

\begin{theorem}\label{Thm: optimality}
For all $t \geq 0$, $ \varepsilon \in (0,a)$ where $ a\coloneqq \min\left\{\frac{9t}{100\pi}, \frac{1}{10}\right\}$, and $d \in \{2,3, \ldots\}$, the following lower bound holds for the sample complexity of sample-based Hamiltonian simulation:
\begin{equation}\label{eq: lower bound of sample complexity}
n_{d}^{\ast}( t,\varepsilon) \geq \frac{32}{1000}\left(\frac{t^{2}}{\varepsilon}\right).
\end{equation}
\end{theorem}

The lower bound in Theorem~\ref{Thm: optimality} has the same scaling behavior as our sample complexity upper bound mentioned in Corollary~\ref{cor:samp-comp-upper-bnd} in terms of $t^2/\varepsilon$, although they have different multiplicative factors (4 versus $\frac{32}{1000}$). 
Hence, the lower bound in Theorem~\ref{Thm: optimality} matches the upper bound up to a multiplicative factor, thus establishing the sample complexity optimality of DME.

To show the lower bound on sample complexity, our strategy is to employ the notion of zero-error query complexity, which we define as the minimum number of copies required to \textit{perfectly} distinguish the ideal unitary operation $e^{i\kappa t}$ for different quantum states $\kappa$. 
Before providing the proof of Theorem~\ref{Thm: optimality}, we first introduce the notion of zero-error query complexity and investigate a general lower bound on the sample complexity of  sample-based Hamiltonian simulation, in terms of this quantity. 
We then recover the lower bound on sample complexity in terms of $t$ and $\varepsilon$ as stated in~\eqref{eq: lower bound of sample complexity}.

\subsection{Lower bound on the sample complexity of sample-based Hamiltonian simulation in terms of zero-error query complexity}

In the following, we develop a sample complexity lower bound on sample-based Hamiltonian simulation, by employing the notion  of zero-error query complexity.
We first define zero-error query complexity as the number of queries needed to  distinguish two ideal unitary channels perfectly. 
Specifically, let $\mathcal{U}_{\kappa,t}(\cdot) =  e^{-i\kappa t}(\cdot)e^{i\kappa t}$ be the ideal unitary evolution according to a quantum state $\kappa$, as we previously defined in~\eqref{eq: ideal evolution by t}.
Then, the zero-error query complexity $m^{\ast}$ can be defined as follows:

\begin{definition}[Zero-error query complexity]
Let $\rho, \sigma \in \mathcal{D}(\mathcal{H}^{d})$ be arbitrary $d$-dimensional quantum states such that $\rho \neq \sigma$. 
Then, we define the integer $m^{\ast}\equiv m^{\ast}(\rho,\sigma,t)\in\mathbb{N}$ as the minimum number of queries needed such that the two unitary channels $\mathcal{U}_{\rho,t}$ and $\mathcal{U}_{\sigma,t}$ for $t>0$ are perfectly distinguishable from each other. That is,
\begin{equation}\label{eq: m-ast-definition}
m^{\ast}(\rho,\sigma,t)\coloneqq \inf\left\{  m\in\mathbb{N}:\frac{1}
{2}\left\Vert \mathcal{U}_{\rho,t}^{\otimes m}-\mathcal{U}_{\sigma,t}^{\otimes
m}\right\Vert _{\diamond}=1\right\}  .
\end{equation}
\end{definition}

It is worth emphasizing that, for arbitrary $\mathcal{U}_{\rho,t}$ and $\mathcal{U}_{\sigma,t}$ such that $t>0$ and $\rho\neq \sigma$, the value $m^{\ast}$ is finite, as shown in Ref.~\cite{acin2001statistical}.

Hereafter, we derive a lower bound on the sample complexity of sample-based Hamiltonian simulation, using the zero-error query complexity $m^{\ast}$ that we have just defined.

\begin{lemma}\label{lem: lower bound of sample complexity}
Let $\rho, \sigma \in \mathcal{D}(\mathcal{H}^{d})$ be arbitrary quantum states such that $\rho \neq \sigma$.
Then, provided that $ m^{\ast}(\rho,\sigma,t)\varepsilon \leq \frac{1}{2}$, we have the following lower bound on the sample complexity $n_{d}^{\ast}(t,\varepsilon)$:
\begin{equation}\label{eq: lower bound of n^ast}
n_{d}^{\ast}(t,\varepsilon) \geq\frac{-\ln\!\left[ 4\, m^{\ast}(\rho,\sigma,t) \varepsilon (1 - m^{\ast}(\rho,\sigma,t) \varepsilon )\right]  }{m^{\ast}(\rho,\sigma,t) \left[ - \ln F(\rho,\sigma) \right]},
\end{equation} 
where $F(\rho, \sigma)$ is the quantum fidelity between $\rho$ and $\sigma$, defined in~\eqref{eq: fidelity definition}.
\end{lemma}

\begin{proof}[Proof of Lemma~\ref{lem: lower bound of sample complexity}]
We first loosen the condition of $n_{d}^{\ast}(t,\varepsilon)$ in~\eqref{eq:sample-complexity-ham-sim}. 
Specifically, instead of considering all possible quantum states $\kappa \in \mathcal{D}(\mathcal{H}^{d})$ on the right-hand side of~\eqref{eq:sample-complexity-ham-sim}, we only consider two fixed quantum states $\rho, \sigma \in \mathcal{D}(\mathcal{H}^{d})$ such that $\rho \neq \sigma$.
For these states, let us define $n^{\ast}(\rho,\sigma,t,\varepsilon)$ to be equal to the minimum number of samples of these states needed to realize a channel that is $\varepsilon$-distinguishable in normalized diamond distance from the ideal unitary evolution:
\begin{equation}
n^{\ast}_d(\rho,\sigma,t,\varepsilon)\coloneqq\inf_{\mathcal{P}^{(n)}
\in\text{CPTP}}\left\{
\begin{array}
[c]{c}
n\in\mathbb{N}:\frac{1}{2}\left\Vert \mathcal{P}^{(n)}\circ\mathcal{A}
_{\rho^{\otimes n}}-\mathcal{U}_{\rho,t}\right\Vert _{\diamond}\leq
\varepsilon,\\
\qquad\quad\,\frac{1}{2}\left\Vert \mathcal{P}^{(n)}\circ\mathcal{A}_{\sigma^{\otimes n}
}-\mathcal{U}_{\sigma,t}\right\Vert _{\diamond}\leq\varepsilon
\end{array}
\right\}  . \label{eq:sample-complexity-pair-ham-sim}
\end{equation}
Then, for all quantum states $\rho, \sigma \in \mathcal{D}(\mathcal{H}^{d})$, we have that
\begin{equation}\label{eq: relaxation}
n_{d}^{\ast}(t,\varepsilon)\geq n_d^{\ast}(\rho,\sigma,t,\varepsilon),
\end{equation}
because the right-hand side of~\eqref{eq:sample-complexity-pair-ham-sim} is a relaxation of the right-hand side of~\eqref{eq:sample-complexity-ham-sim}. 
Accordingly, our revised goal here is to find a lower bound on $n_d^{\ast}(\rho,\sigma,t,\varepsilon)$ in terms of $m^{\ast}(\rho,\sigma,t)$, which will in turn serve as a lower bound on $n_{d}^{\ast}(t,\varepsilon)$ by the inequality in~\eqref{eq: relaxation}.

We now consider $m^{\ast}$ parallel calls of the quantum channels introduced in~\eqref{eq:sample-complexity-pair-ham-sim}, where $m^{\ast} \equiv m^{\ast}(\rho,\sigma,t)$ is the zero-error query complexity.
More specifically, for $n \equiv n_d^{\ast}(\rho,\sigma,t,\varepsilon)$, and for $\mathcal{P}^{(n)}$ being the channel achieving the minimum value on the right-hand side of~\eqref{eq:sample-complexity-pair-ham-sim}, by using the subadditivity of the diamond distance in Lemma~\ref{lem:subadd-DD}, we conclude that
\begin{equation}\label{eq: dia-dis-using-subadd}
\frac{1}{2}\left\Vert \mathcal{P}^{(n)\otimes m^{\ast}}
\circ\mathcal{A}_{\kappa^{\otimes nm^{\ast}}}-\mathcal{U}_{\kappa,t}^{\otimes
m^{\ast}}\right\Vert _{\diamond}\leq 
\frac{m^{\ast}}{2}\left\Vert \mathcal{P}^{(n)}\circ\mathcal{A}_{\kappa^{\otimes n}
}-\mathcal{U}_{\kappa,t}\right\Vert _{\diamond}
\leq m^{\ast}\varepsilon,
\end{equation}
for $\kappa\in\left\{  \rho,\sigma\right\}  $. 
Then the trace distance between $\rho^{\otimes nm^{\ast}}$ and $\sigma^{\otimes nm^{\ast}}$ is bounded from below as follows
\begin{align}
  \frac{1}{2}\left\Vert \rho^{\otimes nm^{\ast}}-\sigma^{\otimes nm^{\ast}
}\right\Vert _{1}
&  \geq\frac{1}{2}\left\Vert \mathcal{P}^{(n)\otimes m^{\ast}}\circ
\mathcal{A}_{\rho^{\otimes nm^{\ast}}}-\mathcal{P}^{(n)\otimes m^{\ast}}
\circ\mathcal{A}_{\sigma^{\otimes nm^{\ast}}}\right\Vert _{\diamond} \label{eq: 48} \\
&  \geq\frac{1}{2}\left\Vert \mathcal{U}_{\rho,t}^{\otimes m^{\ast}
}-\mathcal{U}_{\sigma,t}^{\otimes m^{\ast}}\right\Vert _{\diamond}-\frac{1}
{2}\left\Vert \mathcal{P}^{(n)\otimes m^{\ast}}\circ\mathcal{A}_{\rho^{\otimes
nm^{\ast}}}-\mathcal{U}_{\rho,t}^{\otimes m^{\ast}}\right\Vert _{\diamond
}\nonumber\\
&  \qquad -\frac{1}{2}\left\Vert \mathcal{P}^{(n)\otimes m^{\ast}}
\circ\mathcal{A}_{\sigma^{\otimes nm^{\ast}}}-\mathcal{U}_{\sigma,t}^{\otimes
m^{\ast}}\right\Vert _{\diamond} \label{eq: 49} \\
&  \geq1-2m^{\ast}\varepsilon,
\end{align}
where we used the data-processing inequality in~\eqref{eq: 48}. Specifically, we are using the fact that the trace distance does not increase under the actions of tensoring in an arbitrary state $\omega$ and applying the channel $\mathcal{P}^{(n)\otimes m^{\ast}}$, so that
\begin{equation}
    \left\Vert \rho^{\otimes nm^{\ast}}-\sigma^{\otimes nm^{\ast}
}\right\Vert _{1} \geq \left\Vert \mathcal{P}^{(n)\otimes m^{\ast}}(\rho^{\otimes nm^{\ast}} \otimes \omega) -\mathcal{P}^{(n)\otimes m^{\ast}}(\sigma^{\otimes nm^{\ast}
} \otimes \omega) \right\Vert _{1}.
\end{equation}
This inequality holds for every state $\omega$, so then we can take the supremum over all such states, apply the definition of diamond distance, and arrive at the claimed inequality in~\eqref{eq: 48}. 
We also used the triangle inequality in~\eqref{eq: 49} and the relations given in~\eqref{eq: m-ast-definition} and~\eqref{eq: dia-dis-using-subadd} at the end. 
Therefore, we now have 
\begin{equation}\label{eq: sample complexity bound}
\frac{1}{2}\left(  1-\frac{1}{2}\left\Vert \rho^{\otimes nm^{\ast}}
-\sigma^{\otimes nm^{\ast}}\right\Vert _{1}\right)  \leq m^{\ast}\varepsilon.
\end{equation}
Then, assuming that $2m^{\ast}\varepsilon \leq 1 $, we arrive at the desired sample complexity bound:
\begin{equation}\label{eq: final lower bound of n^ast}
n_d^{\ast}(\rho,\sigma,t,\varepsilon)\geq\frac{\ln\!\left[
4m^{\ast}\varepsilon (1 - m^{\ast}\varepsilon) \right]  }{m^{\ast}\ln F(\rho,\sigma)}.
\end{equation}

To see this clearly, let us argue for this bound here. Note that by Fuchs–van de Graaf inequality~\cite{fuchs1999cryptographic}, the trace distance $\frac{1}{2}\left\| \rho - \sigma \right \|_{1}$ between two quantum states can be bounded by using the quantum fidelity $F(\rho, \sigma)$  as follows:
\begin{align}\label{eq: relation between trace distance and fidelity}
    \frac{1}{2}\left\| \rho - \sigma \right \|_{1} \le \sqrt{1 - F(\rho, \sigma)}.
\end{align}
Then, given that $2m^{\ast}\varepsilon \leq 1 $, combining~\eqref{eq: sample complexity bound} and~\eqref{eq: relation between trace distance and fidelity} results in 
\begin{align}\label{eq: trace distance and fidelity}
    F(\rho^{\otimes nm^{\ast}}, \sigma^{\otimes nm^{\ast}}) \leq 4m^{\ast}\varepsilon (1 - m^{\ast}\varepsilon) .
\end{align}
Therefore, taking the logarithm on both sides and using $F(\rho^{\otimes nm^{\ast}}, \sigma^{\otimes nm^{\ast}}) = F(\rho, \sigma)^{nm^{\ast}}$ gives 
\begin{equation}
nm^{\ast}\ln F(\rho, \sigma) \leq \ln[ 4m^{\ast}\varepsilon (1 - m^{\ast}\varepsilon)],
\end{equation} 
which finally gives a lower bound on the sample complexity in terms of the zero-error query complexity $m^{\ast}$, as stated in~\eqref{eq: final lower bound of n^ast}.
Substituting~\eqref{eq: final lower bound of n^ast} into the right-hand side of~\eqref{eq: relaxation} results in~\eqref{eq: lower bound of n^ast}. 
\end{proof}

\subsection{Optimality proof of the sample complexity bound}\label{section: optimality}

We now provide a proof of Theorem~\ref{Thm: optimality}. 
Specifically, in the proof, we recover the sample complexity lower bound given in~\eqref{eq: lower bound of sample complexity} from our sample complexity bound in Lemma~\ref{lem: lower bound of sample complexity}. 

\medskip 

\begin{proof}[Proof of Theorem~\ref{Thm: optimality}]
We first determine an analytic form for $m^{\ast}(\rho,\sigma,t)$ in~\eqref{eq: m-ast-definition} in terms of $\rho$, $\sigma$, and $t$, again under the assumption that $\rho \neq \sigma$.
More specifically, for all $m\in\mathbb{N}$, the normalized diamond distance between the unitary channels $\mathcal{U}_{\rho,t}^{\otimes m}$ and $\mathcal{U}_{\sigma,t}^{\otimes m}$ can be represented as~\cite[Theorem~3.55]{Watrous2018} 
\begin{equation}\label{eq: diamond distance between unitary channels}
\frac{1}{2}\left\Vert \mathcal{U}_{\rho,t}^{\otimes m}-\mathcal{U}_{\sigma
,t}^{\otimes m}\right\Vert _{\diamond}=\sqrt{1-\min_{|\psi^{(m)}\rangle
}\left\vert \langle\psi^{(m)}|\left(  e^{i\sigma t}e^{-i\rho t}\right)
^{\otimes m}|\psi^{(m)}\rangle\right\vert ^{2}}.
\end{equation}
To find $m^{\ast}(\rho,\sigma,\varepsilon,t)$, we thus need to find the minimum value of $m$ that makes the right-hand side of~\eqref{eq: diamond distance between unitary channels} equal to one.
To prove the lower bound stated in Theorem~\ref{Thm: optimality}, note that we have the freedom to choose a specific pair of states $\{\rho, \sigma\}$, and we can thus choose two-dimensional quantum states \textit{embedded} in a $d$-dimensional Hilbert space.
Then, the two quantum states $\rho$, $\sigma$ can be represented as
\begin{align}
    \rho = \frac{1}{2}(I + r\hat{n}\cdot\hat{\sigma}),\quad  \sigma = \frac{1}{2}(I + r'\hat{n}'\cdot\hat{\sigma}), 
\end{align}
where $I$ denotes the $2\times 2$ identity matrix, $\hat{\sigma}$ denotes the Pauli vector, and $r,r' \in [0,1]$ and $\hat{n},\hat{n}'$ denote the length and direction of the Bloch vectors of $\rho$ and $\sigma$, respectively (i.e., $\hat{n},\hat{n}'$ are unit vectors).
Using these representations, we can write the quantum fidelity between $\rho$ and $\sigma$ as~\cite{hubner1992explicit, jozsa1994fidelity} 
\begin{align}\label{eq: fidelity between two qubit states}
    F(\rho,\sigma) = {\rm Tr}[\rho \sigma] +2  \sqrt{\det\sigma \det\rho} = \frac{1}{2}\left( 1 + rr'(\hat{n}\cdot\hat{n}') + \sqrt{(1 - r^2)(1 - r'^2)} \right).
\end{align}
Also, we can write $e^{i\sigma t} e^{-i\rho t}$ as
\begin{align}\label{eq: product of exponent}
    e^{i\sigma t} e^{-i\rho t} = e^{i\theta \hat{n}''\cdot\hat{\sigma}} = \cos(\theta) I + i\sin(\theta) (\hat{n}''\cdot\hat{\sigma})
\end{align}
for some unit vector $\hat{n}''$, and $\theta$ given by
\begin{align}\label{eq: definition for theta}
    \cos\theta = \cos\!\left(\frac{rt}{2}\right)\cos\!\left(\frac{r't}{2}\right) + \sin\!\left(\frac{rt}{2}\right)\sin\!\left(\frac{r't}{2}\right)(\hat{n}\cdot\hat{n}').
\end{align}

Importantly, since the eigenvalues of $\hat{n}''\cdot\hat{\sigma}$ are $\pm 1$, it follows that the two non-zero eigenvalues of $e^{i\sigma t} e^{-i\rho t}$ in~\eqref{eq: product of exponent} are $e^{\pm i\theta}$, while the rest of the $d-2$ eigenvalues are equal to one. Then, one can determine that the eigenvalues of the tensor-power operator $\left(e^{i\sigma t} e^{-i\rho t}\right)^{\otimes m} $ are $\left\{ e^{\pm i\theta m}, e^{\pm i\theta (m-2)}, e^{\pm i\theta (m-4)}, \ldots, e^{\pm i\theta (m\!\mod 2)}, 1\right\}$.

Next, observe that the following equality holds:
\begin{equation}\label{eq: minimize}
    \min_{|\psi^{(m)}\rangle
}\left\vert \langle\psi^{(m)}|\left(  e^{i\sigma t}e^{-i\rho t}\right)
^{\otimes m}|\psi^{(m)}\rangle\right\vert ^{2} = \min_{|\psi^{(m)}\rangle} \left| \sum_{k} \mu_{k} \left| \langle \psi^{(m)}|\phi_{k}\rangle \right|^2  \right|^2,
\end{equation}
where $\{\mu_k\}_k$ and $\{|\phi_{k}\rangle\}_k$ are the eigenvalues and eigenvectors of $\left(e^{i\sigma t} e^{-i\rho t}\right)^{\otimes m}$, respectively. The term $\sum_{k} \mu_{k} \left| \langle \psi^{(m)}|\phi_{k}\rangle \right|^2$ in~\eqref{eq: minimize} is a convex combination of complex numbers that are distributed over the unit circle (i.e., the complex numbers are given by $\mu_k$ and the probabilities by $\left| \langle \psi^{(m)}|\phi_{k}\rangle \right|^2$). Hence, the problem reduces to minimizing the convex sum of the complex numbers  $\left\{ e^{\pm i\theta m}, e^{\pm i\theta (m-2)}, e^{\pm i\theta (m-4)}, \ldots, e^{\pm i\theta (m\!\mod 2)}, 1\right\}$ over the unit circle, and one can check that this is equal to zero whenever $\theta m \geq \frac{\pi}{2}$, as argued in Ref.~\cite{acin2001statistical}. This leads to the following condition on $m$ that makes the right-hand side of~\eqref{eq: diamond distance between unitary channels} equal to one:  $m \geq \frac{\pi}{2 \theta}$.

To further simplify the analysis, we consider $\hat{n} = (0,0,1)^{T}$ and $\hat{n}' = (0,0,-1)^{T}$, and $r = r'$. 
Then, from~\eqref{eq: fidelity between two qubit states} and~\eqref{eq: definition for theta}, we have
\begin{align}
    \theta = rt, \quad F(\rho,\sigma) = 1 - r^2 ,
\end{align}
which gives the query complexity $m^{\ast} = \lceil \frac{\pi}{2rt} \rceil$.
Accordingly, given that $m^{\ast}\varepsilon = \lceil \frac{\pi}{2rt} \rceil\varepsilon \leq \frac{1}{2}$, by~\eqref{eq: lower bound of n^ast}, the following holds:
\begin{align}
   n_d^{\ast}(\rho,\sigma,t,\varepsilon)
   \geq   \frac{-\ln\!\left[
    4m^{\ast}\varepsilon (1 - m^{\ast}\varepsilon) \right]  }{m^{\ast} \left[ - \ln F(\rho,\sigma)\right]} 
    = \frac{-\ln\!\left[
    4 \lceil \frac{\pi}{2rt} \rceil\varepsilon (1 -  \lceil \frac{\pi}{2rt} \rceil\varepsilon) \right]  }{- \lceil \frac{\pi}{2rt} \rceil\ln (1-r^2)}
\label{eq:not_final_exp}.
\end{align} 
Here, we set $r = \frac{\pi}{2zt}$, for a yet-to-be-determined positive integer $z \in \mathbb{Z}^{+}$, which gives $\lceil \frac{\pi}{2rt}\rceil = \lceil z \rceil = z$. 
Then, from~\eqref{eq:not_final_exp}, we have
\begin{align}
     n_d^{\ast}(\rho,\sigma,t,\varepsilon)
     \geq\frac{-\ln\!\left[
    4z \varepsilon (1 - z\varepsilon) \right]  }{-z\ln (1-\left(\frac{\pi}{2zt}\right)^2)}. \label{eq: not not final}
\end{align}
To proceed, note that $-\frac{1}{x^2}\ln(1 - x^2) \leq 1.151$ for $0 \leq x \leq \frac{1}{2}$.
Hence, let us assume $ \frac{\pi}{2zt} \leq \frac{1}{2}$, which implies the condition $z \geq \frac{\pi}{t}$.
Then, from the right-hand side of~\eqref{eq: not not final} we have
\begin{align}
     n_d^{\ast}(\rho,\sigma,t,\varepsilon)
     \geq  \frac{1}{1.151} \frac{-\ln\!\left[
    4z \varepsilon (1 - z\varepsilon) \right]  }{z \left( \frac{\pi}{2zt} \right)^2} 
    = \frac{-4 z\varepsilon\ln(4z\varepsilon(1-z\varepsilon))}{1.151 \pi^2}  \left( \frac{t^2}{\varepsilon}\right) . \label{eq: not final lower bound}
\end{align}
Here, let $\varepsilon \leq \frac{1}{10}$, which automatically satisfies our assumption $m^{\ast}\varepsilon = z\varepsilon \leq y_0 + \varepsilon\leq \frac{1}{2}$ for~\eqref{eq:not_final_exp}.
We now determine $z = \lceil \frac{y_0}{\varepsilon} \rceil$ for $y_0$ satisfying $y_0 + \varepsilon = 0.19$, such that $z\varepsilon = \lceil \frac{y_0}{\varepsilon} \rceil\varepsilon \leq y_0 + \varepsilon = 0.19$.
Also, given that $\varepsilon \leq \frac{1}{10}$, we have $y_0 \geq 0.09$.
Hence, we always have $z \varepsilon \in [0.09,0.19]$.
Notably, the last term $-y\ln(4y(1-y))$ in the right-hand side of~\eqref{eq: not final lower bound} for $y \in [0.09, 0.19]$ is minimized when $y = 0.19$.
Therefore, from the right-hand side of~\eqref{eq: not final lower bound}, we have
\begin{align}
     n_d^{\ast}(\rho,\sigma,t,\varepsilon)
     \geq \frac{-4(0.19) \ln(4(0.19)(1-(0.19))) }{1.151 \pi^2} \left( \frac{t^2}{\varepsilon}\right) \geq 0.032 \left( \frac{t^2}{\varepsilon}\right) ,
\end{align}
thus obtaining the desired bound. 

To summarize the conditions we have assumed during the proof, we have $\varepsilon \leq \frac{1}{10}$ and $z = \lceil \frac{y_0}{\varepsilon} \rceil\geq \frac{\pi}{t}$, where the latter leads to the condition $\varepsilon \leq \frac{9 t}{100\pi}$.
This concludes the proof. 
\end{proof}

\begin{remark}
We now have a lower bound consistent with the bound derived in Ref.~\cite{Kimmel2017}.
It is worth emphasizing that in the proof of Theorem~\ref{Thm: optimality}, a similar argument can be made for other pairs of states $\{\rho, \sigma\}$, but the states we used are sufficient for establishing our fundamental lower bound on the sample complexity of sample-based Hamiltonian simulation.
\end{remark}

\section{Upper bound on the sample complexity of wave matrix Lindbladization}

\label{sec:upper-bound-WML}

In this section, we derive an upper bound on the non-asymptotic sample complexity of sample-based Lindbladian simulation along with an explicit coefficient. The main result can be summarized as follows:
\begin{theorem} \label{Thm: error bound WML}
Suppose that the Lindbladian $\mathcal{L}$ acts nontrivially on a $d$-dimensional quantum system, where $d\in \mathbb{N}$ and $d\geq 2$. Let $t\geq0$, let $n\in\mathbb{N}$ be such that $n>2dt$, and let $\Delta\coloneqq\frac{t}{n}$. 
For every Lindbladian $\mathcal{L}$ satisfying $\left\|\mathcal{L}\right\|_\diamond\leq 2$, the error of wave matrix Lindbladization satisfies the following bound:
\begin{equation}
\frac{1}{2} \left\Vert e^{\mathcal{L}t}-\left(  \widetilde{e^{\mathcal{L}\Delta}}\right)
^{n}\right\Vert _{\diamond}\leq\frac{3t^{2}d^{2}}{n}.\label{eq:WML-err-bound}%
\end{equation}
where $e^{\mathcal{L}t}$ is defined in~\eqref{eq: ideal lindblad evolution by t} and $\widetilde{e^{\mathcal{L}t}}$ in~\eqref{eq:WML-actual-one-step}. 
\end{theorem}

\medskip

From the error bound of the WML algorithm, we immediately conclude the following sample complexity lower bound:
\begin{corollary}[Upper bound on the sample complexity of WML]
\label{cor:samp-comp-upper-bnd-WML}
Theorem~\ref{Thm: error bound WML} implies  the following upper bound on the sample complexity of sample-based Lindbladian simulation:
\begin{equation}
    n_{d}^{\ast}(t,\varepsilon) \leq  3 d^2  \left( \frac{t^2}{\varepsilon} \right).
\end{equation}
\end{corollary}
\begin{proof}[Proof of Theorem~\ref{Thm: error bound WML}]
By taking $\Delta = \frac{t}{n}$, consider the following chain of inequalities:
\begin{align}
\frac{1}{2} \left\Vert e^{\mathcal{L}t}-\left(  \widetilde{e^{\mathcal{L}\Delta}}\right)
^{n}\right\Vert _{\diamond} &  = \frac{1}{2} \left\Vert \left(  e^{\mathcal{L}\Delta
}\right)  ^{n}-\left(  \widetilde{e^{\mathcal{L}\Delta}}\right)
^{n}\right\Vert _{\diamond}\\
&  \leq n \cdot \frac{1}{2}  \left\Vert e^{\mathcal{L}\Delta}-\widetilde{e^{\mathcal{L}\Delta}%
}\right\Vert _{\diamond}\\
&  \leq n\left( 3\Delta^{2}d^{2} \right)  \\
&  =\frac{3t^{2}d^{2}}{n},
\end{align}
where we have inductively applied Lemma~\ref{lem:subadd-DD}\ to obtain the first inequality. The second inequality follows from
Lemma~\ref{lem:single-step-WML-bound}, which then implies
\eqref{eq:WML-err-bound} after substituting $\Delta=\frac{t}{n}$.
\end{proof}

\begin{lemma}
\label{lem:single-step-WML-bound}
Suppose that $\mathcal{L}$ acts on $d$-dimensional Hilbert space, where $d\in\mathbb{N}$ and $d\geq 2$. Let $t\geq0$, and let $n\in\mathbb{N}$ be such that
$n>2dt$, and set $\Delta\coloneqq\frac{t}{n}$, so that $\Delta<\frac{1}{2d}$.
For the channels $e^{\mathcal{L}\Delta}$ and $\widetilde{e^{\mathcal{L}\Delta
}}$ defined in~\eqref{eq: ideal lindblad evolution by t} with time interval $\Delta$ and
\eqref{eq:WML-actual-one-step}, respectively, the following inequality holds:%
\begin{equation}
\frac{1}{2}\left\Vert e^{\mathcal{L}\Delta}-\widetilde{e^{\mathcal{L}\Delta}}\right\Vert
_{\diamond}\leq 3\Delta^{2}d^{2}.
\end{equation}
\end{lemma}

\begin{proof}
Let us label the $d$-dimensional input system, on which the Lindbladian acts nontrivially as 1, and let us label the reference system as 0. Let us label the $d\times d$-dimensional program system by $2$ and $ 3$, which is prepared in the program state $\psi_L$. We take a series expansion with respect to $\Delta$, which leads to
\begin{align}
\widetilde{e^{\mathcal{L}\Delta}}(\rho)  & =\operatorname{Tr}_{23}[\left(
\operatorname{id}\otimes e^{\mathcal{M}\Delta}\right)  (\rho\otimes\psi
_{L})]\\
& =\rho+\left(  \operatorname{id}\otimes\mathcal{L}\right)  (\rho)\Delta
+\sum_{k=2}^{\infty}\operatorname{Tr}_{23}[\left(  \operatorname{id}%
\otimes\mathcal{M}^{k}\right)  (\rho\otimes\psi_{L})]\frac{\Delta^{k}}{k!},\label{eq:tmp-92}\\
e^{\mathcal{L}\Delta}(\rho)  & =\rho+\left(  \operatorname{id}\otimes
\mathcal{L}\right)  (\rho)\Delta+\sum_{k=2}^{\infty}\left(  \operatorname{id}%
\otimes\mathcal{L}\right)  ^{k}(\rho)\frac{\Delta^{k}}{k!}.
\end{align}
That the first-order term of~\eqref{eq:tmp-92} is equal to $\left(  \operatorname{id}\otimes\mathcal{L}\right)  (\rho)$ was proved in Ref.~\cite{patel2023wave1}. The key idea for this proof is to verify that $\operatorname{Tr}[M(\rho\otimes\psi_L)M^\dagger]$ and $\operatorname{Tr}[M^\dagger M(\rho\otimes\psi_L)]$ are equal to $L\rho L^\dagger$ and $L^\dagger L \rho$, respectively.
Exploiting the triangle inequality for the trace norm, we find that
\begin{align}
\left\Vert e^{\mathcal{L}\Delta}(\sigma_{01})-\widetilde{e^{\mathcal{L}%
\Delta}}(\sigma_{01})\right\Vert _{1}
&=\left\Vert \sum_{k=2}^{\infty}\left(  \mathcal{L}^{k}(\sigma_{01}%
)-\operatorname{Tr}_{23}[\mathcal{M}^{k}(\sigma_{01}\otimes\psi_{L})]\right)
\frac{\Delta^{k}}{k!}\right\Vert _{1}\\
& \leq\sum_{k=2}^{\infty}\left\Vert \mathcal{L}^{k}(\sigma_{01}%
)-\operatorname{Tr}_{23}[\mathcal{M}^{k}(\sigma_{01}\otimes\psi_{L}%
)]\right\Vert _{1}\frac{\Delta^{k}}{k!}\\
& \leq\sum_{k=2}^{\infty}\left(  \left\Vert \mathcal{L}^{k}(\sigma
_{01})\right\Vert _{1}+\left\Vert \operatorname{Tr}_{23}[\mathcal{M}%
^{k}(\sigma_{01}\otimes\psi_{L})]\right\Vert _{1}\right)  \frac{\Delta^{k}%
}{k!}\\
& \leq\sum_{k=2}^{\infty}\left(  \left\Vert \mathcal{L}^{k}(\sigma
_{01})\right\Vert _{1}+\left\Vert \mathcal{M}^{k}(\sigma_{01}\otimes\psi
_{L})\right\Vert _{1}\right)  \frac{\Delta^{k}}{k!}\\
& \leq\sum_{k=2}^{\infty}\left(  \left\Vert \mathcal{L}^{k}\right\Vert
_{\diamond}+\left\Vert \mathcal{M}^{k}\right\Vert _{\diamond}\right)
\frac{\Delta^{k}}{k!}\\
& \leq\sum_{k=2}^{\infty}\left(  \left\Vert \mathcal{L}\right\Vert _{\diamond
}^{k}+\left\Vert \mathcal{M}\right\Vert _{\diamond}^{k}\right)  \frac
{\Delta^{k}}{k!}.\label{eq:tmp_100}
\end{align}
Here, $\left\|\mathcal{L}\right\|_\diamond$ and $\left\|\mathcal{M}\right\|_\diamond$ are bounded from above by $2$ and $2d$, respectively (see Appendix~\ref{appendix:lind-dnorm} for a proof). Therefore, continuing from \eqref{eq:tmp_100},
\begin{align}
\left\Vert e^{\mathcal{L}\Delta}(\sigma_{01})-\widetilde{e^{\mathcal{L}\Delta
}}(\sigma_{01})\right\Vert _{1}  & \leq\sum_{k=2}^{\infty}\left(  \left\Vert
\mathcal{L}\right\Vert _{\diamond}^{k}+\left\Vert \mathcal{M}\right\Vert
_{\diamond}^{k}\right)  \frac{\Delta^{k}}{k!}.\\
& \leq\sum_{k=2}^{\infty}\left(  2^{k}+\left(  2d\right)  ^{k}\right)
\frac{\Delta^{k}}{k!}\\
& \leq2\sum_{k=2}^{\infty}\left(  2d\right)  ^{k}\frac{\Delta^{k}}{k!}
\end{align}
By the assumption $\Delta<\frac{1}{2d}$, it follows that $2d\Delta<1$,
which in turn implies%
\begin{equation}
2\sum_{k=2}^{\infty}\frac{\left(  2d\Delta\right)  ^{k}}{k!}\leq2\left(
\frac{3}{4}\right)  \left(2d\Delta \right)^2=6\Delta^{2}d^{2},
\end{equation}
The first inequality follows because $\sum_{k=2}^\infty \frac{x^k}{k!} = e^x - 1 -x \leq \frac{3}{4}x^2 $ for all $x\in[0,1]$.
This concludes the proof.
\end{proof}

\section{Lower bound on the sample complexity of Lindbladian simulation}

\label{sec:lower-bound-SBLS}

In this section, we derive a lower bound on the sample complexity of sample-based Lindbladian simulation, thus solving a question left open in~\cite{patel2023wave1,patel2023wave2}.
The upper bound on the sample complexity was argued to be $O(t^2/\varepsilon)$ in~\cite{patel2023wave1,patel2023wave2} (however, see the discussion in Appendix~\ref{appendix: incomplete}). 
In fact, we explicitly establish a matching lower bound on the sample complexity for a fixed dimension $d=2$.\\

\subsection{Single-Operator Lindbladian}

Let us first consider open quantum dynamics described by the Lindblad equation with a single Lindblad operator $L$, as given in~\eqref{eq:lindblad_dynamics}. We then have the following lower bound on the sample complexity:

\begin{theorem}[Lower bound for single-operator Lindbladian]
\label{thm:lower_bound_Lindbladian}
Consider an arbitrary sample-based protocol for simulating the Lindbladian evolution $e^{ \mathcal{L} t}$ with a single Lindblad operator $L$ and using the program state $\ket{\psi_L}=(L\otimes I)\ket{\Gamma}$. For every time $t \geq 0$ and every normalized diamond-distance error  $\varepsilon \in (0 ,  \min\{0.039, 0.013 t\}]$, such a protocol requires at least
\begin{equation}
    n_d^*(t,\varepsilon) \geq 10^{-4}\left( \frac{t^2}{\varepsilon} \right)
\end{equation}
copies of the program state in order to achieve an error $\varepsilon$.
\end{theorem}

\begin{proof}
Similar to the Hamiltonian simulation case in Theorem~\ref{Thm: optimality}, we consider two different Lindblad dynamics for a two-dimensional system, i.e., $d=2$, and with the Lindblad operators having the following parameterization:
\begin{equation}
    L_\varphi = \frac{1}{\sqrt{2}} \left( \begin{matrix}
        1 & 0 \\
        0 & e^{i\varphi}
    \end{matrix}\right).
    \label{eq:lower-bnd-WML-choice}
\end{equation}
We take $\varphi=0$ for one of the dynamics, which is a trivial dynamics, ${\cal L}_{L_0}(\rho) = 0$. To specify the Lindbladian dynamics with a given Lindblad operator $L$, we define
\begin{equation}
    \mathcal{L}_L(\rho) \coloneqq L \rho L^\dagger - \frac{1}{2} \left\{L^\dagger L, \rho \right\},
\end{equation}
and the corresponding Lindbladian dynamics as $e^{{\cal L}_L t }$ for time $t$.

We then apply Lemma~\ref{lem: lower_bound_m_copy} below, which generalizes Lemma~\ref{lem: lower bound of sample complexity}.
 We note that when taking $m = m^*(\rho, \sigma, t)$ corresponding to the zero-error query complexity defined in~\eqref{eq: m-ast-definition}, the lower bound reduces to that of Lemma~\ref{lem: lower bound of sample complexity}. However, Lemma~\ref{lem: lower_bound_m_copy} is more general as it can be applied for general $\nu_m$, which is not necessarily equal to one.

Now we evaluate the lower bound in~\eqref{eq:lower_bound_m_copy} by taking $\rho = \psi_{L_\varphi} = \ket{\psi_{L_\varphi}}\! \bra{\psi_{L_\varphi}}$ and $\sigma = \psi_{L_0} = \ket{\psi_{L_0}} \!\bra{\psi_{L_0}}= \frac{1}{2} \ket{\Gamma}\!\bra{\Gamma}$, with the target dynamics ${\cal E}_{\rho,t} = e^{{\cal L}_{L_\varphi} t}$ and ${\cal E}_{\sigma,t} = e^{{\cal L}_{L_0} t} = \mathrm{id}$. By noting that the normalized diamond distance involves the supremum over all possible input states, it follows that
\begin{equation}
    \nu_m \geq \frac{1}{2} \left\| \left[ e^{{\cal L}_{L_\varphi} t}\right]^{\otimes m} (\tau) - \tau \right\|_1,
\end{equation}
for every $m$-qubit density matrix $\tau$. Now let us set $\tau$ to be the Greenberger–Horne–Zeilinger state 
\begin{equation}
    \tau = \frac{1}{2} \left( \ket{0}\!\bra{0}^{\otimes m} + \ket{0}\!\bra{1}^{\otimes m} + \ket{1}\!\bra{0}^{\otimes m} + \ket{1}\!\bra{1}^{\otimes m} \right).
\end{equation}
Observing that $\mathcal{L}_{L_\varphi}(X) = L_\varphi X L_\varphi^\dag - \frac{1}{2} X$ for a general $2\times 2 $ matrix $X$, we note that the superoperator $\mathcal{L}_{L_\varphi}$ has the  following action:
\begin{align}
    {\cal L}_{L_\varphi} (\ket{0}\!\bra{0}) &= 0 ,\\
    {\cal L}_{L_\varphi} (\ket{0}\!\bra{1}) &= \frac{1}{2}\left( e^{-i\varphi} -1 \right) \ket{0}\!\bra{1} ,\\
    {\cal L}_{L_\varphi} (\ket{1}\!\bra{0}) &= \frac{1}{2}\left( e^{i\varphi} -1 \right) \ket{1}\!\bra{0} ,\\
    {\cal L}_{L_\varphi} (\ket{1}\!\bra{1}) &= 0.
\end{align}
This leads to the evolution of the density matrix elements corresponding to $\ket{0}\!\bra{1}$ and $\ket{1}\!\bra{0}$ as $\rho_{01}(t) = e^{\frac{1}{2}\left( e^{-i\varphi} -1 \right) t } \rho_{01}(0)$ and $\rho_{10}(t) = e^{\frac{1}{2}\left( e^{i\varphi} -1 \right) t } \rho_{10}(0)$, respectively. Hence, the $m$-repetition of these dynamics is as follows:
\begin{equation}
    \left[ e^{ {\cal L}_{L_\varphi} t} \right]^{\otimes m} (\tau) = \frac{1}{2} \left( \ket{0}\!\bra{0}^{\otimes m} + e^{\frac{1}{2}\left( e^{-i\varphi} -1 \right) mt }\ket{0}\!\bra{1}^{\otimes m} + e^{\frac{1}{2}\left( e^{i\varphi} -1 \right) mt }\ket{1}\!\bra{0}^{\otimes m} + \ket{1}\!\bra{1}^{\otimes m} \right).
\end{equation}
As the diagonal elements do not change, we obtain
\begin{align}
    &\frac{1}{2} \left\| \left[ e^{{\cal L}_{L_\varphi} t}\right]^{\otimes m} (\tau) - \tau \right\|_1 \notag \\
    &= \frac{1}{4} \left\| (e^{\frac{1}{2}\left( e^{-i\varphi} -1 \right) mt }-1)\ket{0}\!\bra{1}^{\otimes m} + (e^{\frac{1}{2}\left( e^{i\varphi} -1 \right) mt }-1)\ket{1}\!\bra{0}^{\otimes m} \right\|_1\\
    &= \frac{1}{2} \left[ 1 + e^{-mt (1-\cos\varphi)} - 2 e^{-\frac{mt}{2}(1-\cos\varphi)} \cos\!\left(\frac{mt}{2}\sin\varphi\right)\right]^{\frac{1}{2}}.
\end{align}
Now, we take the parameter $\varphi$ to satisfy $\sin\varphi = \frac{2\pi}{mt}$, which leads to 
\begin{equation}
    \nu_m \geq \frac{1}{2} \left| 1 + e^{-\frac{mt}{2} (1-\cos\varphi)}\right|
    \geq \frac{1}{2}.
\end{equation}
Meanwhile, the fidelity between the program states can be directly calculated as
\begin{equation}
    F(\psi_{L_\varphi}, \psi_{L_0}) = |\bra{\Gamma} (L_0^\dagger \otimes \mathbb{I}) (L_\varphi \otimes \mathbb{I}) \ket{\Gamma}|^2 = \frac{1+\cos\varphi}{2} = \frac{1 + \sqrt{1 - \left( \frac{2\pi}{mt} \right)^2}}{2}.
\end{equation}
The lower bound on the sample complexity in \eqref{eq:lower_bound_m_copy} then becomes
\begin{align}
    n_d^*(t,\varepsilon) &\geq \frac{-\ln\!\left[ 1- \left(\frac{1}{2}- 2m \varepsilon \right)^2 \right]}{m \left[-\ln \!\left( \frac{1 + \sqrt{1 - \left( \frac{2\pi}{mt} \right)^2}}{2} \right) \right]} \\
    &\geq \frac{-\ln\!\left[ 1- \left(\frac{1}{2} - 2m \varepsilon \right)^2 \right]}{m \left(\ln 2\right) \left( \frac{2\pi}{mt} \right)^2} \\
    &= \left( \frac{1}{4\pi^2 \ln 2} \right) \left(\frac{t^2}{\varepsilon}\right) (m\varepsilon) \left[  -\ln\!\left[ 1- \left(\frac{1}{2}  - 2m \varepsilon \right)^2 \right]\right],
\end{align}
from the fact that $-\frac{1}{x^2}\ln \!\left(\frac{ 1 + \sqrt{1 - x^2}}{2} \right) \leq \ln 2$ for $0 \leq x \leq 1$. Finally, by taking $m = \lfloor \alpha/\varepsilon \rfloor$ and $\varepsilon \leq \frac{\alpha}{2}$, which implies $ \frac{1}{2} \alpha \leq \left( 1- \frac{\varepsilon}{\alpha}\right) \alpha\leq m\varepsilon \leq \alpha$, we can bound the last term as
\begin{equation}
\label{eq:bound_lindbladian_ft}
        (m\varepsilon) \left[  -\ln\!\left[ 1- \left(\frac{1}{2}  - 2m \varepsilon \right)^2 \right]\right] \geq \frac{\alpha}{2} \left[  -\ln\!\left[ 1- \left(\frac{1}{2}  - 2 \alpha \right)^2 \right]\right],
\end{equation}
where the right-hand side of the inequality achieves the maximum value $\approx 0.0049$ by taking $\alpha = \alpha^* \approx 0.08$. By combining all of these, we obtain the lower bound of sample complexity as
\begin{equation}
    n_d^*(t,\varepsilon) \geq \left( \frac{0.0049}{4\pi^2 \ln 2} \right) \left(\frac{t^2}{\varepsilon}\right) \geq 0.00018 \left(\frac{t^2}{\varepsilon}\right),
\end{equation}
under the conditions $ \frac{2\pi \epsilon}{\alpha^* t} \leq \frac{2 \pi}{mt} \leq 1 \Rightarrow  \frac{\varepsilon}{t} \leq \frac{\alpha^*}{2\pi} \leq 0.013$ and $\varepsilon \leq \frac{\alpha^*}{2} \leq 0.039$.
\end{proof}

\begin{lemma} \label{lem: lower_bound_m_copy}Let $\rho, \sigma \in \mathcal{D}(\mathcal{H}^{d})$ be arbitrary quantum states such that $\rho \neq \sigma$. For all $m \in \mathbb{N}$, consider the normalized diamond distance between $m$ queries to the ideal channels ${\cal E}_{\rho,t}$ and ${\cal E}_{\sigma,t}$, where these are of the form stated just below~\eqref{eq:channel_error_definition}:
\begin{equation}
    \nu_m(\rho,\sigma,t) \coloneqq  \frac{1}{2} \left\| {\cal E}_{\rho,t}^{\otimes m} - {\cal E}_{\sigma,t}^{\otimes m} \right\|_\diamond.
\end{equation}    
Then, provided that $ \nu_m(\rho,\sigma,t) \geq 2m \varepsilon$, we have the following lower bound on the sample complexity $n_{d}^{\ast}(t,\varepsilon)$:
    \begin{equation} \label{eq:lower_bound_m_copy}
        n^*_d(t, \varepsilon) \geq \frac{-\ln\!\left[ 1- (\nu_m(\rho,\sigma,t) - 2m \varepsilon)^2 \right]}{m \left[-\ln F(\rho,\sigma) \right]}\, .
    \end{equation}
\end{lemma}

\begin{proof}
See Appendix~\ref{appendix:another_lower_bound}.
\end{proof}

\subsection{General Lindbladian}

Building on the lower bound established in the previous section for sample-based simulation of single-operator Lindbladians (Theorem~\ref{thm:lower_bound_Lindbladian}), along with the lower bound for sample-based Hamiltonian simulation (Theorem~\ref{Thm: optimality}), we now derive lower bounds for sample-based simulation of general Lindbladians, that is, Lindbladians composed of both multiple Lindblad operators and a Hamiltonian term. Specifically, we consider Lindbladians of the form
\begin{equation}\label{eq:lindbladmaster-gen}
    \mathcal{L}(\rho) \coloneqq -i[H, \rho] + \sum_{k=1}^{K} \left( L_k \rho L_k^{\dagger} - \frac{1}{2} \left\{ L_k^{\dagger} L_k, \rho \right\} \right).
\end{equation}

We assume that the Hamiltonian $H$ is given as a linear combination of program state chosen from the set $\{ \sigma_j \}_{j=1}^{J}$:
\begin{equation}
    H \coloneqq \sum_{j=1}^{J} c_j \sigma_j,
\end{equation}
where $c_j \in \mathbb{R}$. Each Lindblad operator $L_k$ is assumed to be a local operator acting on a constant number of qubits, and is encoded into a pure state $|\psi_k\rangle$ as
\begin{equation}
    |\psi_k\rangle \coloneqq \frac{(L_k \otimes I) |\Gamma\rangle}{\left\| L_k \right\|_2}.
    \label{eq:psi_k-L_k-def}
\end{equation}

In Ref.~\cite{patel2023wave2}, the authors proposed an algorithm that uses
\begin{equation}
    n_j = O\!\left( \frac{|c_j| c t^2}{\varepsilon} \right)
\end{equation}
copies of the program state $\sigma_j$, and
\begin{equation}
    m_k = O\!\left( \frac{\| L_k \|_2^2 c t^2}{\varepsilon} \right)
\end{equation}
copies of the program state $|\psi_k\rangle$, in order to approximately simulate the quantum channel $e^{\mathcal{L}t}$ to within normalized diamond-distance error $\varepsilon \in [0, 1]$. Here $c$ is defined as follows:
\begin{equation}
    c\coloneqq \sum_{j=1}^{J}\left\vert c_{j}\right\vert +\sum_{k=1}^{K}\left\Vert
L_{k}\right\Vert_{2}^{2}\label{eq:def-c}.
\end{equation}
This means that the total number of program states we need is
\begin{equation}
    \sum_j n_j + \sum_k m_k = O\!\left( \frac{c^2 t^2}{\varepsilon} \right).
\end{equation}
Note that this holds only for Lindbladians with local Lindblad operators that act nontrivially on a constant number of qubits. For general Lindbladians, the sample complexity may scale with the system dimension $d$, as seen in the single-operator case. For a more detailed derivation of the above equation, refer to~\cite[Appendix~D.1]{Sims2025}.

In what follows, we establish a lower bound for simulating general Lindbladians using sample-based protocols with local Lindblad operators, and we show that the WML algorithm of Ref.~\cite{patel2023wave2} is asymptotically optimal.

\begin{theorem}[Lower bound for general Lindbladian]
\label{thm:gen-lind-lower-bnd}
Any sample-based protocol that approximately simulates the evolution $e^{\mathcal{L}t}$ generated by a general Lindbladian of the form in~\eqref{eq:lindbladmaster-gen}, using program states chosen from the set $\{\sigma_j\}_{j=1}^{J}$ for the Hamiltonian $H$ and $\{|\psi_k\rangle\}_{k=1}^{K}$ for the Lindblad operators, must use at least
\begin{equation}
    \Omega\!\left( \frac{c^2 t^2}{\varepsilon} \right)
\end{equation}
copies of these program states to achieve normalized diamond-distance error at most $\varepsilon \in (0, 1]$, where $c$ is defined as in~\eqref{eq:def-c}. This lower bound holds under the assumption that each Lindblad operator $L_k$ acts nontrivially on a constant number of qubits.
\end{theorem}
\begin{proof}
 From the definition of $c$ (see~\eqref{eq:def-c}), one of the following three cases must hold:
 \begin{enumerate}
     \item $\sum_{j: c_j > 0} c_j \geq c/3$,
     \item $\sum_{j: c_j < 0} |c_j| \geq c/3$, or
     \item $\sum_{k=1}^{K} \left\|L_k\right\|_2^2 \geq c/3$.
 \end{enumerate}

If Case~1 holds, we take $\sigma_j = \sigma$ for all $j$ such that $c_j > 0$, take $\sigma_j = I/\sqrt{d}$ for all~$j$ with $c_j < 0$, and take $L_k = \left\|L_k\right\|_2 \cdot I/\sqrt{2}$. In this case, the Lindbladian simplifies to $\mathcal{L}(\rho) = \sum_{j: c_j > 0} -i[\sigma, \rho]$. Simulating this evolution requires $\Omega((\sum_{j: c_j > 0} c_j)^2 t^2 / \varepsilon)$ samples of $\sigma$. Since Case~1 assumes that $\sum_{j: c_j > 0} c_j \geq c/3$, this implies that we need $\Omega(c^2 t^2/\varepsilon)$ samples of $\sigma$.

Case~2 can be treated analogously.

For Case~3, we take $\sigma_j = I/\sqrt{d}$ for all $j$, and take $L_k = \left\|L_k\right\|_2 \cdot L/\sqrt{2}$, where $L$ is a fixed local Lindblad operator satisfying $\left\|L\right\|_2 = 1$. In this setup, each $|\psi_k\rangle$ becomes the same state $|\psi_L\rangle \coloneqq (L \otimes I)|\Gamma\rangle$. In this case, the Lindbladian simplifies to $\mathcal{L}(\rho) = \sum_k \left\|L_k\right\|_2^2 ( L \rho L^{\dagger} - 1/2 \{ L^{\dagger} L, \rho\})$. Simulating this Lindbladian requires $\Omega((\sum_k \left\|L_k\right\|_2^2)^2 t^2 / \varepsilon)$ samples of $|\psi_L\rangle$. Since Case~3 assumes that $\sum_{k=1}^{K} \left\|L_k\right\|_2^2 \geq c/3$, this implies that we need $\Omega(c^2 t^2/\varepsilon)$ samples of $|\psi_L\rangle$.

Thus, in all three cases, the total number of program states required is at least $\Omega(c^2 t^2 / \varepsilon)$. 
\end{proof}

\medskip
Let us finally note that a matching upper bound, under the same assumptions as stated in Theorem~\ref{thm:gen-lind-lower-bnd}, was established in \cite[Appendix~D.1]{Sims2025}.

\section{Concluding remarks}

\label{sec:conclusion}

In this work, we established a detailed sample complexity analysis of sample-based Hamiltonian and Lindbladian simulation.
We first derived an error bound for DME and WML in terms of the number of samples and evolution time. 
More precisely, we showed that the sample complexity $n$ to achieve a desired imprecision level $\varepsilon$ in the normalized diamond distance for evolution time $t$ is no larger than $4t^2/\varepsilon$ for DME, and $3d^2t^2/\varepsilon$ for WML.
We also examined a fundamental lower bound on the sample complexity of sample-based Hamiltonian and Lindbladian simulations, by exploiting the zero-error query complexity, which is the minimum number of queries to two unknown unitary channels such that they are near-perfectly distinguishable from each other.
We found that given program states as two-dimensional states embedded in a $d$-dimensional Hilbert space, the fundamental lower bound we have derived shows the optimality of 
DME up to a multiplicative factor. We also found that from the particular example presented in \eqref{eq:lower-bnd-WML-choice}, the fundamental lower bound proves optimality of WML up to a dimensional factor.

We now mention some open problems. 
The first is to extend our optimality results for sample-based Hamiltonian and Lindbladian simulation, given in Section~\ref{sec:lower-bnd-samp-complex-SBHS} and Section~\ref{sec:lower-bound-SBLS}, to more general program states.
The virtue of DME and WML is that the program state can be arbitrarily chosen, and thus, it holds for any finite-dimensional program state as long as multiple copies of the program state are available. 
Hence, by using different program states, one would obtain tighter sample complexity lower bounds for sample-based Hamiltonian and Lindbladian simulation.

Another interesting question is to improve our sample complexity bounds for sample-based Hamiltonian and Lindbladian simulation. 
More specifically, because the asymptotic scaling behavior of the sample complexity is $\Theta(t^2/\varepsilon)$, a goal is to reduce the gap between the constant prefactors of our sample complexity bounds.
By reflecting back on the proofs we have provided, the constant prefactors of our sample complexity lower bounds might be further improved.

Further extending our sample complexity bounds for different computational tasks beyond Hamiltonian and Lindbladian simulation is also an interesting problem. 
For example, DME was first introduced not for Hamiltonian simulation, but for a subroutine of the quantum principal component analysis~\cite{lloyd2014quantum}. Here, controlled-DME is used for phase encoding applied to eigenstates of the target density matrix. We have verified in Appendix~\ref{appendix:ctrlDMB} that qPCA also shares same the sample complexity upper bound up to a logarithmic factor.

Furthermore, the ideas behind our sample complexity bounds might be further extended to analyze the complexity of quantum recursion algorithms that exponentiate different quantum states, including fixed-point search algorithms~\cite{grover2005fixed, yoder2014fixed}, double-bracket quantum algorithms~\cite{gluza2024double, robbiati2024double, gluza2024doublebraket}, or quantum dynamical programming~\cite{son2025quantum}.

\begin{acknowledgments}
This work was supported by the National Research Foundation of Korea (NRF) Grants No.~RS-2024-00438415. HK is supported by the KIAS Individual Grant No.~CG085302 at Korea Institute for Advanced Study.
DP and MMW acknowledge support from
the Air Force Office of Scientific Research  under agreement no.~FA2386-24-1-4069.    
    
The U.S.~Government is authorized to reproduce and
distribute reprints for Governmental purposes notwithstanding any copyright
notation thereon. The views and conclusions contained herein are those of the
authors and should not be interpreted as necessarily representing the official
policies or endorsements, either expressed or implied, of the United States Air Force.

\end{acknowledgments}

\section*{Author contributions}

\noindent \textbf{Author Contributions}:
The following describes the different contributions of the authors of this work, using roles defined by the CRediT
(Contributor Roles Taxonomy) project~\cite{NISO}:

\medskip 
\noindent \textbf{BG}: Formal analysis, Validation, Writing - Original draft, Writing - Review \& Editing.

\medskip 
\noindent \textbf{HK}: Formal analysis, Methodology, Funding acquisition, Writing - Review \& Editing, Supervision.

\medskip 
\noindent \textbf{SP}: Formal analysis, Validation, Resources, Writing - Review \& Editing.

\medskip 
\noindent \textbf{DP}: Conceptualization, Formal analysis, Investigation, Methodology, Validation, Writing - Review \& Editing.

\medskip 
\noindent \textbf{MMW}: Conceptualization, Formal analysis, Funding acquisition, Methodology, Supervision, Validation, Writing - Original draft,  Writing - Review \& Editing.

\bibliography{ref}

\appendix

\section{Incompleteness of the previous sample complexity bound for DME}

\label{appendix: incomplete}

In this appendix, we argue the incompleteness of the previous sample complexity bound of DME~\cite{lloyd2014quantum, Kimmel2017}.
To do so, we first recall the previous error analysis of DME in the diamond distance norm, which was presented in Ref.~\cite[Supplementary Information Section~A]{Kimmel2017}. We note here that similar arguments presented below, imply incompleteness of the error analysis of WML put forward in \cite{patel2023wave1}. The issues with WML have been resolved in Section~\ref{sec:upper-bound-WML} and \cite[Appendix~D.1]{Sims2025}.

Let $\rho_{RS} \in \mathcal{D}(\mathcal{H}_{R}\otimes\mathcal{H}_{S})$ be an unknown bipartite quantum state over the joint system $RS$, where $S$ is the input system we are interested in, and $R$ is an arbitrary reference system.
Then the unitary evolution of system $S$ according to the ideal unitary channel $\mathcal{U}_{\sigma, \Delta}$ in~\eqref{eq:DME-ideal-one-step} for time $\Delta$ can be represented as
\begin{align}
&\left(\mathcal{I}_{R}\otimes\mathcal{U}_{\sigma, \Delta}\right)(\rho_{RS}) \nonumber \\
&= (\mathbb{I}_{R} \otimes e^{-i\sigma_{S} \Delta}) \rho_{RS} (\mathbb{I}_{R} \otimes e^{i\sigma_{S} \Delta}) \\
&= \rho_{RS} - i[(\mathbb{I}_{R} \otimes \sigma_{S}), \rho_{RS}]\Delta - \frac{1}{2!}[ (\mathbb{I}_{R}\otimes\sigma_{S}) , [(\mathbb{I}_{R}\otimes \sigma_{S}) , \rho_{RS}]]\Delta^2 + \cdots. \label{apx: eq: ideal evolution}
\end{align}

To implement DME, we prepare $n$ copies of the program state $\sigma$ in the ancillary system $A_k$, where $k$ ranges from $1$ to $n$.
For each copy of $\sigma$, the DME algorithm simply applies the Hamiltonian $\operatorname{SWAP}$ between systems $S$ and $A_k$ for time step $\Delta = \frac{t}{n}$, and then discards $A_k$.
Here, $\operatorname{SWAP}_{SA_k} \coloneqq\sum_{i,j}|i\rangle\!\langle j|_{S}\otimes|j\rangle\!\langle i|_{A_k}$ is the swap operator between system $S$ and $A_k$.
Given that the $\operatorname{SWAP}$ operator is self-inverse (i.e., $\operatorname{SWAP}^2 = \mathbb{I}$), the evolution by $\operatorname{SWAP}$ Hamiltonian for time $\Delta$ can be represented as
\begin{align}
e^{-i\Delta \operatorname{SWAP}_{SA_k}} = \cos\Delta\cdot\mathbb{I}_{SA_k} -i\sin\Delta\cdot\operatorname{SWAP}_{SA_k}.
\end{align}
Using this convention, the state after the first iteration of the above procedure can be explicitly written as
\begin{align}
&\text{Tr}_{A_1}\left[( \mathbb{I}_{R} \otimes  e^{-i\operatorname{SWAP}_{SA_1} \Delta} )(\rho_{RS} \otimes \sigma_{A_1})( \mathbb{I}_{R} \otimes e^{i\operatorname{SWAP}_{SA_1} \Delta}  )\right] \nonumber \\
&= \rho_{RS} \cos^2 \Delta - i[(\mathbb{I}_{R} \otimes \sigma_{S} ), \rho_{RS}] \sin \Delta \cos \Delta +  \text{Tr}_{S}(\rho_{RS}) \otimes  \sigma_{S}\sin^2 \Delta, \\
&= \rho_{RS} - i[( \mathbb{I}_{R} \otimes \sigma_{S} ), \rho_{RS}] \Delta - (\rho_{RS} -  \text{Tr}_{S}(\rho_{RS}) \otimes \sigma_{S}) \Delta^2 + O(\Delta^3). \label{apx: eq: first approximation}
\end{align}
Here, the difference in trace distance between the ideal state in~\eqref{apx: eq: ideal evolution} and the first approximation in~\eqref{apx: eq: first approximation} is
\begin{align}
\frac{1}{2} \left\| (\mathbb{I}_{R} \otimes e^{-i\sigma_{S} \Delta}) \rho_{RS} (\mathbb{I}_{R} \otimes e^{i\sigma_{S} \Delta}) - \text{Tr}_{A_1}\left[( \mathbb{I}_{R} \otimes  e^{-iS_{SA_1} \Delta} )(\rho_{RS} \otimes \sigma_{A_1})( \mathbb{I}_{R} \otimes e^{iS_{AA_1} \Delta}  )\right] \right\|_1 \leq O(\Delta^2). 
\end{align}
If we denote by $\tilde{\rho}_{RS}^{[k]}$ the state after $k$ iterations of this procedure (so $\tilde{\rho}_{RS}^{[0]} = \rho_{RS}$ is the original state and $\tilde{\rho}_{RS}^{[1]}$ is the state in~\eqref{apx: eq: first approximation}), we get the following recursion:
\begin{align}
\tilde{\rho}_{RS}^{[k]} = \tilde{\rho}_{RS}^{[k-1]} - i[(\mathbb{I}_{R} \otimes \sigma_{S}), \tilde{\rho}_{RS}^{[k-1]}] \Delta - (\tilde{\rho}_{RS}^{[k-1]} - \text{Tr}_{S}(\tilde{\rho}_{RS}^{[k-1]}) \otimes \sigma_{S}) \Delta^2 + O(\Delta^3). 
\end{align}
By evaluating this recursively, the final state after $n$ iterations of the procedure can be expressed as
\begin{align}
\tilde{\rho}_{RS}^{[n]} &= \tilde{\rho}_{RS}^{[n-m]} - i[(\mathbb{I}_{R} \otimes \sigma_{S}), \tilde{\rho}_{RS}^{[n-m]}] m\Delta - (\tilde{\rho}_{RS}^{[n-m]} -\text{Tr}_{S}(\tilde{\rho}_{RS}^{[n-m]}) \otimes  \sigma_{S}) m\Delta^2 \nonumber \\
&\quad + i[ (\mathbb{I}_{R}\otimes\sigma_{S}) , i[(\mathbb{I}_{R}\otimes \sigma_{S}) , \tilde{\rho}_{RS}^{[n-m]}]] (1 + 2 + \cdots + (m-1)) \Delta^2 + O(\Delta^3),\label{apx: eq: DME-final-any m}
\end{align}
for all $m \in \{0, \dots, n\}$. (Note that~\cite[Eq.~(A7)]{Kimmel2017}, in the paper's supplemental information, features a typo, where $(1 + 2 + \cdots + m)$ should instead be $(1 + 2 + \cdots + (m-1))$, as written above.) In particular, for $m = n$ we get
\begin{align}
\tilde{\rho}_{RS}^{[n]} &= \rho_{RS} - i[(\mathbb{I}_{R} \otimes \sigma_{S}), \rho_{RS}] n\Delta - (\rho_{RS} -  \text{Tr}_{S}(\rho_{RS}) \otimes \sigma_{S}) n\Delta^2 \nonumber \\
&\quad + i[ (\mathbb{I}_{R}\otimes \sigma_{S}), i[(\mathbb{I}_{R}\otimes \sigma_{S} ), \rho_{RS}]] \frac{n(n-1)}{2} \Delta^2 + O(\Delta^3). \label{apx: eq: DME-final}
\end{align}
On the other hand, the desired final state at time $t = n\Delta$ is given by~\eqref{apx: eq: ideal evolution} with $n\Delta$ instead of $\Delta$, so that
\begin{multline}
\left(\mathcal{I}_{R}\otimes\mathcal{U}_{\sigma, t}\right)(\rho_{RS})  \\
= \rho_{RS} - i[(\mathbb{I}_{R} \otimes \sigma_{S}), \rho_{RS}]n\Delta - \frac{1}{2!}[ (\mathbb{I}_{R}\otimes\sigma_{S}) , [(\mathbb{I}_{R}\otimes \sigma_{S}) , \rho_{RS}]]n^2\Delta^2 + \cdots. \label{apx: eq: ideal final evolution}
\end{multline}
Hence, comparing~\eqref{apx: eq: DME-final} and~\eqref{apx: eq: ideal final evolution}, at first glance it seems that the error induced by the whole process of DME is bounded by $O(n \Delta^2) = O(t^2/n)$.

To sum up, the proof shows that each of the zeroth, first, and second-order terms in~\eqref{apx: eq: ideal final evolution} are canceled by each of the corresponding terms in~\eqref{apx: eq: DME-final}, such that only $O(n\Delta^2)$ term in~\eqref{apx: eq: DME-final} remains up to the second order.
However, this proof excludes the case that $t$ is not asymptotically small.
That is, the higher-order terms on the right-hand side of~\eqref{apx: eq: ideal final evolution} cannot be neglected unless the evolution time $t$ asymptotically converges to $0$. 
For example, if the $n^3\Delta^3$ term in~\eqref{apx: eq: ideal final evolution} is not canceled by the corresponding term (i.e., $O(\Delta^3)$ term) in~\eqref{apx: eq: DME-final}, the final imprecision would scale as $n^3\Delta^3 = t^3$, which cannot be arbitrarily reduced by increasing the sample number $n$.

Therefore, the proof itself guarantees the imprecision bound $O(t^2/n)$ only when $t$ asymptotically converges to $0$. 
To conclude that the imprecision is bounded by $O(t^2/n)$ for an arbitrary $t$, it should be promised that each of the higher order terms $\Delta^3$, $\Delta^4$, $\cdots$ in~\eqref{apx: eq: DME-final} cancels each of the higher order terms $n^3\Delta^3$, $n^4\Delta^4$, $\cdots$ in~\eqref{apx: eq: ideal final evolution}, respectively. However, this has not been argued in~\cite[Supplementary Information Section~A]{Kimmel2017}.

\section{Subadditivity of diamond distance}

\begin{lemma}
[Subadditivity of diamond distance]\label{lem:subadd-DD}Let $\mathcal{N}_{1}$,
$\mathcal{N}_{2}$, $\mathcal{M}_{1}$, and $\mathcal{M}_{2}$ be channels. Then
\begin{equation}
\left\Vert \mathcal{N}_{2}\circ\mathcal{N}_{1}-\mathcal{M}_{2}\circ
\mathcal{M}_{1}\right\Vert _{\diamond}\leq\left\Vert \mathcal{N}
_{1}-\mathcal{M}_{1}\right\Vert _{\diamond}+\left\Vert \mathcal{N}
_{2}-\mathcal{M}_{2}\right\Vert _{\diamond}.
\end{equation}

\end{lemma}

\begin{proof}
Let $\rho_{RA}  \in \mathcal{D}(\mathcal{H}_{R}\otimes\mathcal{H}_{A})$ be an arbitrary bipartite state.
\ Then, we have
\begin{align}
&  \left\Vert \left(  \operatorname{id}_{R}\otimes\left(  \mathcal{N}_{2}
\circ\mathcal{N}_{1}\right)  \right)  \left(  \rho_{RA}\right)  -\left(
\operatorname{id}_{R}\otimes\left(  \mathcal{M}_{2}\circ\mathcal{M}
_{1}\right)  \right)  \left(  \rho_{RA}\right)  \right\Vert _{1}\nonumber\\
&  =\left\Vert
\begin{array}
[c]{c}
\left(  \operatorname{id}_{R}\otimes\left(  \mathcal{N}_{2}\circ
\mathcal{N}_{1}\right)  \right)  \left(  \rho_{RA}\right)  -\left(
\operatorname{id}_{R}\otimes\left(  \mathcal{N}_{2}\circ\mathcal{M}
_{1}\right)  \right)  \left(  \rho_{RA}\right) \\
+\left(  \operatorname{id}_{R}\otimes\left(  \mathcal{N}_{2}\circ
\mathcal{M}_{1}\right)  \right)  \left(  \rho_{RA}\right)  -\left(
\operatorname{id}_{R}\otimes\left(  \mathcal{M}_{2}\circ\mathcal{M}
_{1}\right)  \right)  \left(  \rho_{RA}\right)
\end{array}
\right\Vert _{1}\\
&  \leq\left\Vert \left(  \operatorname{id}_{R}\otimes\left(  \mathcal{N}
_{2}\circ\mathcal{N}_{1}\right)  \right)  \left(  \rho_{RA}\right)  -\left(
\operatorname{id}_{R}\otimes\left(  \mathcal{N}_{2}\circ\mathcal{M}
_{1}\right)  \right)  \left(  \rho_{RA}\right)  \right\Vert _{1}\nonumber\\
&  \qquad+\left\Vert \left(  \operatorname{id}_{R}\otimes\left(
\mathcal{N}_{2}\circ\mathcal{M}_{1}\right)  \right)  \left(  \rho_{RA}\right)
-\left(  \operatorname{id}_{R}\otimes\left(  \mathcal{M}_{2}\circ
\mathcal{M}_{1}\right)  \right)  \left(  \rho_{RA}\right)  \right\Vert _{1}\\
&  \leq\left\Vert \left(  \operatorname{id}_{R}\otimes\mathcal{N}_{1}\right)
\left(  \rho_{RA}\right)  -\left(  \operatorname{id}_{R}\otimes\mathcal{N}
_{2}\right)  \left(  \rho_{RA}\right)  \right\Vert _{1}+\left\Vert
\mathcal{N}_{2}-\mathcal{M}_{2}\right\Vert _{\diamond}\\
&  \leq\left\Vert \mathcal{N}_{1}-\mathcal{M}_{1}\right\Vert _{\diamond
}+\left\Vert \mathcal{N}_{2}-\mathcal{M}_{2}\right\Vert _{\diamond}.
\end{align}
Here, the first inequality follows from the data-processing inequality for the trace
distance under the channel $\mathcal{N}_{2}$.
Also, the second inequality holds because the state $\left(  \operatorname{id}_{R}\otimes\mathcal{M}
_{1}\right)  \left(  \rho_{RA}\right)  $ is a particular state in $\mathcal{D}(\mathcal{H}_{R}\otimes\mathcal{H}_{A})$ 
to consider for
the optimization of the diamond distance norm $\left\Vert \mathcal{N}_{2}-\mathcal{M}_{2}\right\Vert
_{\diamond}$, whose optimization is over all input states in $\mathcal{D}(\mathcal{H}_{R}\otimes\mathcal{H}_{A})$. 
Therefore, since the inequality we have derived
\begin{equation}
\left\Vert \left(  \operatorname{id}_{R}\otimes\left(  \mathcal{N}_{2}
\circ\mathcal{N}_{1}\right)  \right)  \left(  \rho_{RA}\right)  -\left(
\operatorname{id}_{R}\otimes\left(  \mathcal{M}_{2}\circ\mathcal{M}
_{1}\right)  \right)  \left(  \rho_{RA}\right)  \right\Vert _{1}
\leq\left\Vert \mathcal{N}_{1}-\mathcal{M}_{1}\right\Vert _{\diamond
}+\left\Vert \mathcal{N}_{2}-\mathcal{M}_{2}\right\Vert _{\diamond}
\end{equation}
holds for every input state $\rho_{RA}$, we conclude the desired statement.
\end{proof}

\section{Proof of Lemma~\ref{lem: lower_bound_m_copy}}

\label{appendix:another_lower_bound}

This section presents a sample complexity bound that generalizes the bound from Lemma~\ref{lem: lower bound of sample complexity}, which uses the $m$-copy discrimination task stated therein,  but it can also be applied to a general target channel ${\cal E}_{\sigma, t}$ which is realized in terms of the program state $\sigma$. Similar to the proof of Lemma~\ref{lem: lower bound of sample complexity}, we start by defining 
\begin{equation} \label{eq:appd_n_d^*}
n^{\ast}_d(\rho,\sigma,t,\varepsilon)\coloneqq\inf_{\mathcal{P}^{(n)}
\in\text{CPTP}}\left\{
\begin{array}
[c]{c}
n\in\mathbb{N}:\frac{1}{2}\left\Vert \mathcal{P}^{(n)}\circ\mathcal{A}
_{\rho^{\otimes n}}-\mathcal{E}_{\rho,t}\right\Vert _{\diamond}\leq
\varepsilon,\\
\qquad\quad\,\frac{1}{2}\left\Vert \mathcal{P}^{(n)}\circ\mathcal{A}_{\sigma^{\otimes n}
}-\mathcal{E}_{\sigma,t}\right\Vert _{\diamond}\leq\varepsilon
\end{array}
\right\},
\end{equation}
where we have that
\begin{equation}
n_{d}^{\ast}(t,\varepsilon)\geq n_d^{\ast}(\rho,\sigma,t,\varepsilon).
\end{equation}
We now consider $m$ parallel calls of the quantum channels with $n \equiv n_d^{\ast}(\rho,\sigma,t,\varepsilon)$ and $\mathcal{P}^{(n)}$ being a channel satisfying the constraints given in~\eqref{eq:appd_n_d^*}. Then by using the subadditivity of the diamond distance from Lemma~\ref{lem:subadd-DD}, we obtain
\begin{equation}
\frac{1}{2}\left\Vert \mathcal{P}^{(n)\otimes m}
\circ\mathcal{A}_{\kappa^{\otimes nm}}-\mathcal{E}_{\kappa,t}^{\otimes
m}\right\Vert _{\diamond}\leq 
\frac{m}{2}\left\Vert \mathcal{P}^{(n)}\circ\mathcal{A}_{\kappa^{\otimes n}
}-\mathcal{E}_{\kappa,t}\right\Vert _{\diamond}
\leq m\varepsilon,
\end{equation}
for $\kappa\in\left\{  \rho,\sigma\right\}$. The trace distance between $\rho^{\otimes nm}$ and $\sigma^{\otimes nm}$ is then bounded from below as follows:
\begin{align}
  \frac{1}{2}\left\Vert \rho^{\otimes nm}-\sigma^{\otimes nm}\right\Vert _{1}
&  \geq\frac{1}{2}\left\Vert \mathcal{P}^{(n)\otimes m}\circ
\mathcal{A}_{\rho^{\otimes nm}}-\mathcal{P}^{(n)\otimes m}
\circ\mathcal{A}_{\sigma^{\otimes nm}}\right\Vert _{\diamond} \\
&  \geq\frac{1}{2}\left\Vert \mathcal{E}_{\rho,t}^{\otimes m
}-\mathcal{E}_{\sigma,t}^{\otimes m}\right\Vert _{\diamond}-\frac{1}
{2}\left\Vert \mathcal{P}^{(n)\otimes m}\circ\mathcal{A}_{\rho^{\otimes
nm}}-\mathcal{E}_{\rho,t}^{\otimes m}\right\Vert _{\diamond
}\nonumber\\
&  \qquad -\frac{1}{2}\left\Vert \mathcal{P}^{(n)\otimes m}
\circ\mathcal{A}_{\sigma^{\otimes nm}}-\mathcal{E}_{\sigma,t}^{\otimes
m}\right\Vert _{\diamond} \\
&  \geq \nu_m(\rho,\sigma,t) -2m\varepsilon,
\end{align}
following the same logic as in the proof of Lemma~\ref{lem: lower bound of sample complexity}. The only difference is that now we allow a non-zero error for the channel discrimination task; i.e.,
\begin{equation}
    \nu_m(\rho, \sigma,t)= \frac{1}{2}\left\Vert \mathcal{E}_{\rho,t}^{\otimes m}-\mathcal{E}_{\sigma,t}^{\otimes m}\right\Vert _{\diamond}
\end{equation}
is not necessarily equal to one. By using the relation in~\eqref{eq: relation between trace distance and fidelity} we obtain
\begin{align}
    F(\rho,\sigma)^{nm} \leq 1 - (\nu_m(\rho,\sigma,t) - 2m\varepsilon)^2,
\end{align}
provided that $\nu_m(\rho,\sigma,t)\geq 2m\varepsilon$. This finally leads to
\begin{equation}
\label{eq:upper_bound_m_new}
    n_d^*(t,\varepsilon)\geq n_d^*(\rho,\sigma,t,\varepsilon) \geq \frac{-\ln\!\left[ 1- (\nu_m(\rho,\sigma,t) - 2m \varepsilon)^2 \right]}{m \left[-\ln F(\rho,\sigma) \right]}\, .
\end{equation} 

\section{Diamond norms of Lindbladian operators}

\label{appendix:lind-dnorm}

In this section, we provide a proof that $\|\mathcal{L}\|_\diamond\leq2$ if $\left \| L\right\|_2 =1$ and $\|\mathcal{M}\|_\diamond\leq2d$ for $\mathcal{M}$ defined in~\eqref{eq:WML-M-op}. Let $\tau$ be an arbitrary bipartite quantum state of a reference system $R$ and an input system $S$. Then,
\begin{align}
 \left\Vert (\operatorname{id} \otimes  \mathcal{L})(\tau)\right\Vert _{1}
& =\left\Vert (I\otimes L)\tau (I\otimes L^{\dag})-\frac{1}{2}\left\{  I\otimes L^{\dag}L,\tau\right\}  \right\Vert _{1}\\
& \leq \left\Vert (I\otimes L)\tau (I\otimes L^{\dag})\right\Vert _{1}+\frac{1}{2}\left\Vert
(I\otimes L^{\dag}L)\tau\right\Vert _{1}+\frac{1}{2}\left\Vert \tau (I\otimes L^{\dag
}L)\right\Vert _{1}\\
& =\left\Vert \tau^{1/2}(I\otimes L^{\dag}L)\tau^{1/2}\right\Vert
_{1}+\left\Vert (I\otimes L^{\dag}L)\tau\right\Vert _{1}\\
& =\operatorname{Tr}\!\left[ \tau^{1/2}(I\otimes L^{\dag}L)\tau^{1/2}\right]+\left\Vert (I\otimes L^{\dag}L)\tau\right\Vert _{1}\\
& \leq\left\Vert (I\otimes L^{\dag}L)\right\Vert _{\infty}\left\Vert \tau
^{1/2}\tau^{1/2}\right\Vert _{1}+\left\Vert (I\otimes L^{\dag
}L)\right\Vert _{\infty}\left\Vert \tau\right\Vert
_{1}\\
& =2\left\Vert L^{\dag}L\right\Vert _{\infty} \\
&\leq 2\left\Vert L^{\dag}L\right\Vert _{1} \\
&= 2 ,
\end{align}
where the last equality follows from the assumption that the Lindblad operator has a unit Hilbert--Schmidt norm. Therefore we conclude that
\begin{equation}
\left\|\mathcal{L}\right\|_{\diamond}=\sup_{\tau}\left\|(\operatorname{id}\otimes \mathcal{L})(\tau)\right \|_1\leq2.
\end{equation}

Similarly,
\begin{align}
\left\Vert (\operatorname{id} \otimes \mathcal{M}) (\tau )\right\Vert _{1}
& =\left\Vert (I\otimes M)\tau (I\otimes M^{\dag})-\frac{1}{2}\left\{  I\otimes M^{\dag}%
M,\tau \right\}  \right\Vert _{1}\\
& \leq\left\Vert (I\otimes M)\tau (I\otimes M^{\dag})\right\Vert _{1}+\frac{1}%
{2}\left\Vert (I\otimes  M^{\dag}M)\tau \right\Vert _{1} +\frac{1}{2}\left\Vert \tau (I\otimes M^{\dag}M)\right\Vert _{1}\\
& =\left\Vert \tau ^{1/2}(I\otimes M^{\dag}M)\tau %
^{1/2}\right\Vert _{1}+\left\Vert (I\otimes M^{\dag}M) \tau\right\Vert _{1} \\
& =\operatorname{Tr}\!\left[ \tau ^{1/2}(I\otimes M^{\dag}M)\tau 
^{1/2}\right]+\left\Vert (I\otimes M^{\dag}M) \tau\right\Vert _{1} \\
& \leq\left\Vert I\otimes M^{\dag}M\right\Vert _{\infty}\left\Vert \tau^{1/2}\tau ^{1/2}\right\Vert _{1}+\left\Vert
I\otimes M^{\dag}M\right\Vert _{\infty}\left\Vert \tau \right\Vert
_{1}\\
& = 2\left\Vert M^{\dag}M\right\Vert _{\infty} \\
& = 2d.
\end{align}
The last equality follows from the fact that
\begin{align}
M M^{\dag}  & = I_1 \otimes |\Gamma\rangle\!\langle\Gamma|_{23}  =I_{1} \otimes d|\Phi\rangle\!\langle\Phi|_{23} ,
\end{align}
so that
\begin{equation}
\left\Vert M^{\dag}M\right\Vert _{\infty}= \left\Vert M M^{\dag}\right\Vert _{\infty} = d.
\end{equation}
Therefore we conclude that
\begin{equation}
\left\Vert \mathcal{M}\right\Vert _{\diamond} = \sup_{\tau}\|\operatorname{id} \otimes \mathcal{M}\|_1\leq2d,
\end{equation}
thus completing the proof. 

\section{Density matrix exponentiation as subroutine of an algorithm}

\label{appendix:ctrlDMB}

Density matrix exponentiation can be utilized beyond Hamiltonian simulation. The ability to exponentiate an arbitrary density matrix in polynomial time can be a powerful subroutine combined with the quantum phase estimation algorithm. As an example, we verify the non-asymptotic sample complexity of a practical quantum algorithm, quantum principal component analysis (qPCA)~\cite{lloyd2014quantum}. The key idea of qPCA is to encode eigenvalues of a program state to phase with a sequence of controlled-DME. That is, given multiple copies of program state $\rho=\sum_{i}r_i \ket{\chi_i}\!\bra{\chi_i}$, the algorithm returns $\sum_i r_i \ket{\Tilde{r_i}} \!\bra{\Tilde{r_i}} \otimes \ket{\chi_i} \!\bra{\chi_i}$ where $\Tilde{r_i}$ is $T$-bit estimation of eigenvalue $r_i$. Principal components can be statistically obtained by measuring the ancillary register, given that the rank of $\rho$ is effectively low. DME is utilized to asymptotically realize a controlled unitary for quantum phase estimation in qPCA.

Here, we demonstrate a non-asymptotic upper bound of controlled-DME on an eigenstate $\chi$ is lower than \eqref{eq:DME-err-bound} as expected, leading to the non-asymptotic upper bound of sample complexity of qPCA, $\Tilde{O}(t^2/\varepsilon)$. The controlled-DME can be viewed as DME with extended program state: $\rho\to\ket{0}\!\bra{0}\otimes\rho$. Therefore, controlled-DME in this instance can be written as the following:
\begin{align}
    &\Tilde{\mathcal{U}}_{\ket{1}\!\bra{1}\otimes\rho, \Delta}(P\otimes\chi) \\ &= \cos^2\Delta P\otimes\chi - i\frac{\sin 2\Delta}{2} [\ket{1}\!\bra{1}\otimes\rho, P\otimes\chi]+\sin^2\Delta \ket{1}\!\bra{1}\otimes\rho \notag \\
    &= \cos^2\Delta P\otimes\chi - i\frac{\sin 2\Delta}{2}r [\ket{1}\!\bra{1}, P]\otimes \chi+\sin^2\Delta \ket{1}\!\bra{1}\otimes\rho
\end{align} for any $P\in\{\sigma_x, \sigma_y, \sigma_z, \mathbb{I}\}$. Here, we have abused notation to drop subscript $i$: $r=r_i, \chi=|\chi_i\rangle\langle\chi_i|$.
Noting Pauli anti-commutations $[\frac{\mathbb{I}-\sigma_z}{2}, \sigma_x] = -i\sigma_y$, $[\frac{\mathbb{I}-\sigma_z}{2}, \sigma_y] = i\sigma_x$, we obtain transition rule of the linear map $\mathcal{U}_{\ket{1}\!\bra{1}\otimes\rho, \Delta}$ between finite numbers of \emph{states}:
\begin{align}
    \Tilde{\mathcal{U}}_{\ket{1}\!\bra{1}\otimes\rho, \Delta}(\sigma_x\otimes\chi) & = \cos^2\Delta \sigma_x\otimes\chi -r\frac{\sin 2\Delta}{2} \sigma_y\otimes\chi+\sin^2\Delta \ket{1}\!\bra{1}\otimes\rho, \\
   \Tilde{\mathcal{U}}_{\ket{1}\!\bra{1}\otimes\rho, \Delta}(\sigma_y\otimes\chi) & = \cos^2\Delta \sigma_y\otimes\chi +r\frac{\sin 2\Delta}{2} \sigma_x\otimes\chi+\sin^2\Delta \ket{1}\!\bra{1}\otimes\rho, \\
   \Tilde{\mathcal{U}}_{\ket{1}\!\bra{1}\otimes\rho, \Delta}(\sigma_z\otimes\chi) & = \cos^2\Delta \sigma_z\otimes\chi +\sin^2\Delta \ket{1}\!\bra{1}\otimes\rho, \\
   \Tilde{\mathcal{U}}_{\ket{1}\!\bra{1}\otimes\rho, \Delta}(\mathbb{I}\otimes\chi)&  = \cos^2\Delta \mathbb{I}\otimes\chi +\sin^2\Delta \ket{1}\!\bra{1}\otimes\rho, \\
   \Tilde{\mathcal{U}}_{\ket{1}\!\bra{1}\otimes\rho, \Delta}(\ket{1}\!\bra{1}\otimes\rho) & = \ket{1}\!\bra{1}\otimes\rho.
\end{align}
Therefore, the exact output of controlled-DME can be computed as
\begin{equation}\label{eq:qpca1}
\begin{aligned}
    &\Tilde{\mathcal{U}}_{\ket{1}\!\bra{1}\otimes\rho, t}\left(\gamma\otimes\chi\right) \\&= \frac{1}{2}\cos^{2n}(t/n)
    \begin{bmatrix}
        1+z & (x-iy)(1+ir\tan(t/n))^n \\ (x+iy)(1-ir\tan(t/n))^n & 1-z \\
    \end{bmatrix} 
    \otimes \chi \\
    &\qquad +\left(1-\cos^{2n}(t/n)\right) \ket{1}\!\bra{1}\otimes\rho,
\end{aligned}
\end{equation}
for control qubit $\gamma = \frac{\mathbb{I}+x \sigma_x+y\sigma_y+z\sigma_z}{2}$. In the quantum phase estimation example, $(x, y, z)=(1, 0, 0)$.
The trace distance between Eq.~\eqref{eq:qpca1} and
\begin{equation}
    {\mathcal{U}}_{\ket{1}\!\bra{1}\otimes\rho, t}\left(\frac{\mathbb{I}+\sigma_x}{2}\otimes\chi\right)=\frac{1}{2}
    \begin{bmatrix}
        1 & e^{irt} \\ e^{-irt} & 1 \\
    \end{bmatrix} 
    \otimes \chi,
\end{equation}
is simplified to a qubit representation as
\begin{equation}\label{eq:qpca2}
\begin{aligned}
    &\varepsilon_{\rho,t}(\chi)\\&\coloneq\frac{1}{2}\left\|{[\Tilde{\mathcal{U}}_{\ket{1}\!\bra{1}\otimes\rho, t}-{\mathcal{U}}_{\ket{1}\!\bra{1}\otimes\rho, t}]\left(\frac{\mathbb{I}+\sigma_x}{2}\otimes\chi\right)}\right\|_1 \\
    &= \frac{1}{2} \left\|{\frac{1}{2}
    \begin{bmatrix}
        \cos^{2n}(t/n)-1 & \cos^{n}(t/n)(\cos(t/n)+ir\sin(t/n))^n - e^{irt}\\ \cos^{n}(t/n)(\cos(t/n)-ir\sin(t/n))^n - e^{-irt} & (1-2r)(\cos^{2n}(t/n)-1) \\
    \end{bmatrix} }\right\|_1\\
    &\qquad +\frac{1}{2}(1-r)\left(1-\cos^{2n}(t/n)\right).
\end{aligned}
\end{equation}
It is worth noting that $\cos^{n}(t/n)\geq1-\frac{t^2}{2n}$ and $\cos^{2n}(t/n)\geq1-\frac{t^2}{n}$ for further derivation. To find the upper bound of the trace distance, we first bound the difference in \emph{phase} element.
Let us define
\begin{equation}
    A_{r,t/n}e^{i\phi_{r,t/n}}\coloneq\cos(t/n)+ir\sin(t/n),
\end{equation}
or equivalently,
\begin{equation}
    A_{r,t/n} \coloneq \sqrt{1-(1-r)^2\sin^2(t/n)},~~\phi_{r,t/n}\coloneq\tan^{-1}(r\tan(t/n)),
\end{equation}
such that
\begin{equation}
    \left\lvert{\phi_{r, t/n}-\frac{rt}{n}}\right\rvert\leq\frac{2}{\pi}(1-r)\frac{t^2}{n^2},~~\cos(t/n)\leq A_{r, t/n}\leq 1 ,
\end{equation}
for $n\geq \frac{2t}{\pi}$. We then have
\begin{align}
    &\left\lvert{(e^{-irt/n}\cos(t/n)(\cos(t/n)+ir\sin(t/n)))^n - 1}\right\rvert^2 \notag \\
    & \leq 1+A_{r, t/n}^{2n}\cos^{2n}(t/n)-2A_{r, t/n}^{n}\cos^{n}(t/n)\cos\!\left(\frac{2}{\pi}(1-r)\frac{t^2}{n}\right) \\
    & \leq 1+A_{r, t/n}^{2n}\cos^{2n}(t/n)-2A_{r, t/n}^{n}\cos^{n}(t/n)\left(1-\frac{2}{\pi^2}(1-r)^2\frac{t^4}{n^2}\right) \\
    &=(1-A_{r, t/n}^{n}\cos^{n}(t/n))^2+A_{r, t/n}^{n}\cos^{n}(t/n)\frac{{4}(1-r)^2t^4}{\pi^2n^2} \\
    & \leq (1-\cos^{2n}(t/n))^2+\cos^{n}(t/n)\frac{{4}(1-r)^2t^4}{\pi^2n^2} \\
    & \leq (1-\cos^{2n}(t/n))^2+\left(\frac{{2}(1-r)t^2}{\pi n}\right)^2 \\
    & \leq \left(1+\left[\frac{{2}(1-r)}{\pi}\right]^2\right)^2\left(\frac{t^2}{n}\right)^2 ,
\end{align}
and
\begin{multline}
    \left\lvert{(e^{-irt/n}\cos(t/n)(\cos(t/n)+ir\sin(t/n)))^n - 1}\right\rvert^2
    \\
    \geq(1-A_{r, t/n}^{n}\cos^{n}(t/n))^2
    \geq(1-\cos^n(t/n))^2.
\end{multline}
Since the determinant is negative, the first term of the right-hand side of~\eqref{eq:qpca2} can be bounded by the following: 
\begin{align}
    \|{\cdots}\|_1^2&\leq r^2(1-\cos^{2n}(t/n))^2+4\left(1+\left[\frac{{2}(1-r)}{\pi}\right]^2\right)^2\left(\frac{t^2}{n}\right)^2 \\
    &\leq\left(4\left(1+\left[\frac{{2}(1-r)}{\pi}\right]^2\right)^2+r^2\right)\left(\frac{t^2}{n}\right)^2.
\end{align}
Finally, we obtain upper bound of \eqref{eq:qpca2},
\begin{equation}
    \varepsilon_{\rho,t}(\chi)\leq\left[\frac{1-r}{2}+\sqrt{\left(\frac{{2}(1-r)}{\pi}+1\right)^2+\frac{r^2}{4}}\right]\frac{t^2}{n}\leq\left(\frac{3}{2} + \frac{4}{\pi^2}\right)\frac{t^2}{n},
\end{equation}
which is about $1.91{t^2}/{n}$.

In qPCA, $j$th controlled-DME is applied to eigenstates for $t_j=(2\pi)2^{j-1}$ time at the cost of $m$ program state $\rho$. Therefore, the total trace distance $\varepsilon(\chi)$ of the qPCA output is
\begin{equation}
    \varepsilon(\chi)\leq\sum_{j=1}^{T}\varepsilon_{\rho, t_j}(\chi)\leq\sum_{j=1}^{T}2(2\pi)^2\frac{4^{j-1}}{m}\leq\frac{8\pi^2}{m}\frac{4^{T}-1}{3}.
\end{equation}
The total simulation time $t$ and number of program states $n$ are
\begin{align}
    t &= \sum_{j=1}^T (2\pi)2^{j-1} = 2\pi(2^T-1), \\
    n &= Tm, 
\end{align}
respectively. Therefore, we obtain the final upper bound:
\begin{equation}
    \varepsilon(\rho) \leq \sum_{i} r_i \varepsilon(\chi_i) = \frac{2}{3} \frac{t(t+4\pi)}{n}\log_2t \in \Tilde{O}(t^2/n).
\end{equation}

\end{document}